\documentclass[11pt,a4paper]{article}
\title{\bf Geometry on the
Wasserstein space \\
over a compact Riemannian manifold}
\author{Hao DING$^{1,2}$\footnote{Email:  dinghao16@mails.ucas.ac.cn}
\quad Shizan FANG$^1$\footnote{Email:Shizan.Fang@u-bourgogne.fr}
\vspace{3mm}\\
{\footnotesize $^1$Institut de Math\'ematiques de Bourgogne, UMR 5584 CNRS, }\\
{\footnotesize Universit\'e de Bourgogne Franche-Comt\'e, F-21000 Dijon, France}\\
{\footnotesize $^2$Institute of Applied Mathematics, Academy of Mathematics and Systems Science,}\\
{\footnotesize Chinese Academy of Sciences, Beijing 100190, China}
}
\date{March 29, 2021}
\usepackage{amssymb,amsmath,amsfonts,amsthm,array,bm,color}

\def\N{\mathbb{N}}
\def\R{\mathbb{R}}

\def\P{\mathbb{P}}

\def\TT{\bar{\mathbf{T}}}
\def\bn{\bar{\nabla}}
\def\bD{\bar{D}}
\def\D{{\mathbb D}}

\def\C{\mathcal{C}}

\def\L{\mathcal L}

\def\div{\textup{div}}

\def\eps{\varepsilon}
\def\<{\langle}
\def\>{\rangle}

\let \dis=\displaystyle
\let\ra=\rightarrow

\newtheorem{theorem}{Theorem}[section]
\newtheorem{lemma}[theorem]{Lemma}       

\newtheorem{proposition}[theorem]{Proposition}
\newtheorem{remark}[theorem]{Remark}

\newtheorem{definition}[theorem]{Definition}

\begin{document}

\maketitle
\makeatletter 
\renewcommand\theequation{\thesection.\arabic{equation}}
\@addtoreset{equation}{section}
\makeatother 

\vspace{-8mm}
\begin{abstract} We will revisit the intrinsic differential geometry of the Wasserstein space over a Riemannian manifold, 
due to a series of papers by Otto, Otto-Villani, Lott, Ambrosio-Gigli-Savar\'e and so on. 
\end{abstract}



\section{Introduction}\label{sect1}

\quad  
For the sake of simplicity, we will consider in this paper  a connected compact Riemannian manifold $M$ of dimension $m$.
We denote by  $d_M$ the Riemannian distance and $dx$ the Riemannian measure on $M$ such that $\int_M dx=1$.
Since the diameter of $M$ is finite, 
any probability measure $\mu$ on $M$ is such that $\int_M d_M^2(x_0, x)\, d\mu(x)<+\infty$, where $x_0$ is a fixed point 
of $M$. As usual, we denote by $\P_2(M)$ the space of probability measures on $M$, endowed with  the Wasserstein distance $W_2$ 
defined by

\begin{equation*}
W_2^2(\mu_1, \mu_2) =\inf\Bigl\{ \int_{M\times M} d_M^2(x,y)\,\pi(dx,dy),\quad \pi\in \C(\mu_1, \mu_2)\Bigr\},
\end{equation*}
where $\C(\mu_1, \mu_2)$ is the set of probability measures $\pi$ on $M\times M$, having $\mu_1, \mu_2$ as two 
marginal laws. It is well known that $\P_2(M)$ endowed with $W_2$ is a Polish space. In this compact case, the weak convergence 
for probability measures is metrized by $W_2$; therefore $(\P_2(M), W_2)$ is a compact Polish space.

\vskip 2mm
The introduction of tangent spaces of $\P_2(M)$ can go back to the early work \cite{OV}, as well as in 
\cite{Otto}. A more rigorous treatment was given in \cite{AGS}. In differential geometry, for a smooth curve 
$\{c(t);\ t\in [0,1]\}$ on a manifold $M$, the derivative $c'(t)$ with respect to the time $t$ is in the tangent space :
$\dis c'(t)\in T_{c(t)}M$. A classical result says that for an absolutely continuous curve $\{c(t);\ t\in [0,1]\}$ on  $M$, 
the derivative $c'(t)\in T_{c(t)}M$ exists for almost all $t\in [0, 1]$. Following \cite{AGS}, we say that a curve $\{c(t);\ t\in [0,1]\}$ 
on $\P_2(M)$ is absolutely continuous in $L^2$ if there exists $k\in L^2([0,1])$ such that

\begin{equation*}\label{eq1.1}
W_2(c(t_1), c(t_2))\leq \int_{t_1}^{t_2} k(s)\, ds,\quad t_1<t_2.
\end{equation*}

The following result is our starting point: 

\begin{theorem}[see \cite{AGS}, Theorem 8.3.1] \label{th1}Let $\{c_t; \ t\in [0,1]\}$ 
be an absolutely continuous curve on $\P_2(M)$ in $L^2$, then there exists a Borel vector field $Z_t$ on $M$ 
such that 
\begin{equation*}
\int_{[0,1]}\Bigl[\!\! \int_M |Z_t(x)|_{T_xM}^2\, dc_t(x)\Bigr]\, dt<+\infty
\end{equation*}

and the following continuity equation 
\begin{equation}\label{eq1.2}
\frac{dc_t}{dt}+ \nabla\cdot (Z_t\, c_t)=0,
\end{equation}
holds in the sense of distribution. Uniqueness  to \eqref{eq1.2} holds if moreover $Z_t$ is imposed to be in 
\begin{equation*}
 \overline{\bigl\{ \nabla\psi,\ \psi\in C^\infty(M) \bigr\}}^{L^2(c_t)}.
\end{equation*}
\end{theorem}

\vskip 2mm

In this work, we define the tangent space $\TT_\mu$ of $\P_2(M)$ at $\mu$ by 

\begin{equation}\label{eq1.2-1}
\TT_\mu = \overline{\bigl\{ \nabla\psi,\ \psi\in C^\infty(M) \bigr\}}^{L^2(\mu)},
\end{equation}
the closure of gradients of smooth functions in the space $L^2(\mu)$. Equation \eqref{eq1.2} implies that for almost all $t\in [0,1]$,
\begin{equation}\label{eq1.3}
\frac{d}{dt} \int_M f(x)\, dc_t(x)=\int_M \langle \nabla f(x), Z_t(x)\rangle_{T_xM}\, dc_t(x),\quad f\in C^1(M).
\end{equation}

We will say that $Z_t$ is the intrinsic derivative of $c_t$ and use the notation 
\begin{equation*}
\frac{{d}^I c_t}{dt}=Z_t \in \TT_{c_t}.
\end{equation*}

\vskip 2mm
In what follows, we will describe the tangent space $\TT_\mu$ with the least conditions as possible on the measure $\mu$.
Consider the quadratic form defined by 
\begin{equation*}
{\mathcal E}(\psi)=\int_ M |\nabla\psi(x)|^2\, d\mu(x),\quad \psi\in C^1(M).
\end{equation*}

We assume that there is a constant $C_\mu>0$ such that
\begin{equation}\label{eq1.3.1}
\int_M (\psi -\langle\psi\rangle )^2\, d\mu \leq C_\mu\, \int_M |\nabla\psi|^2\, d\mu,
\end{equation}
where $\dis \langle\psi\rangle=\int_M \psi(x)\, dx$.  The condition \eqref{eq1.3.1} is satisfied if $\mu$ admits a positive 
density $\rho>0$: $d\mu = \rho\, dx$. In fact, let 
\begin{equation*}
\beta_1=\inf_{x\in M} \rho(x) >0, \quad \beta_2=\sup_{x\in M} \rho(x)<+\infty.
\end{equation*}

Since $M$ is compact, the following Poincar\'e inequality holds :
\begin{equation*}
\int_M (\psi -\langle\psi\rangle )^2\, dx \leq C\, \int_M |\nabla\psi|^2\, dx,
\end{equation*}
then
\begin{equation*}
\int_M (\psi -\langle\psi\rangle )^2\, d\mu \leq \frac{C\beta_2}{\beta_1}\, \int_M |\nabla\psi|^2\, d\mu.
\end{equation*}

\vskip 2mm
Now let $Z\in \TT_\mu$; there is a sequence of functions $\psi_n\in C^\infty(M)$ such that 
$\dis Z=\lim_{n\ra +\infty} \nabla\psi_n$ in $L^2(\mu)$. By changing $\psi_n$ to $\psi_n-\langle\psi_n\rangle$
and by condition \eqref{eq1.3.1}, $\{\psi_n;\ n\geq 1\}$ is a Cauchy sequence in $L^2(\mu)$. If the quadratic form 
${\mathcal E}(\psi)$ is closable in $L^2(\mu)$, 
then there exists a function $\varphi_\mu$ in the Sobolev space  $\D_1^2(\mu)$ such that $Z=\nabla\varphi_\mu$, 
where $\D_1^2(\mu)$ is the closure of 
$C^\infty(M)$ with respect to the norm 
\begin{equation*}
||\varphi||_{\D_1^2(\mu)}^2:= \int_M |\varphi(x)|^2\, d\mu(x) + \int_M |\nabla\varphi(x)|^2\, d\mu(x).
\end{equation*}

A sufficient condition to insure the closability for ${\mathcal E}$ is that the formula of integration by parts 
holds  for $\mu$; more precisely, for any $C^1$ vector field $Z$ on $M$, there exists a function 
denoted by $\div_\mu(Z)\in L^2(\mu)$ such that
\begin{equation}\label{eq1.4}
\int_M \langle \nabla f(x), Z(x)\rangle_{T_xM}\, d\mu(x)=-\int_M f(x)\,\div_\mu(Z)(x),\quad f\in C^1(M).
\end{equation}
\vskip 2mm

\begin{definition}\label{def1.1}
We say that  the measure $\mu$ is a {\it measure having divergence}
 if $\div_\mu(Z)\in L^2(\mu)$ exists. We will use the notation
 \begin{equation*}
 \P_{\div}(M)
 \end{equation*}
 to denote the set of probability measures on $M$ having strictly positive continuous density and satisfying conditions \eqref{eq1.4}.
 \end{definition}
\vskip 2mm

\begin{proposition}\label{prop1.1} For a measure $\mu\in\P_{\div}(M)$, we have
\begin{equation*}
\TT_\mu=\bigl\{\nabla\psi;\ \psi\in \D_1^2(\mu) \bigr\}.
\end{equation*}
\end{proposition}

\vskip 2mm
The inconvenient for \eqref{eq1.3} is the existence of derivative for almost all $t\in [0,1]$. 
In what follows, we will present two typical classes of absolutely continuous curves in $\P_2(M)$.
\vskip 2mm

\subsection{\bf Constant vector fields on $\P_2(M)$} \label{section1.1}
For any gradient vector field $\nabla\psi$ on $M$ with $\psi\in C^\infty(M)$, 
consider the ordinary differential equation (ODE):
\begin{equation*}
\frac{d}{dt}U_t(x)=\nabla\psi(U_t(x)),\quad U_0(x)=x\in M.
\end{equation*}
Then $x\ra U_t(x)$ is a flow of diffeomorphisms on $M$. Let $\mu\in\P_2(M)$, consider 
$c_t=(U_t)_\#\mu$. It is easy to see that the curve $\{c_t; \ t\in [0,1]\}$ is absolutely continuous in $L^2$ and for $f\in C^1(M)$, 
\begin{equation*}
\frac{d}{dt}\int_M f(x)\, dc_t(x)=\frac{d}{dt}\int_M f(U_t(x))\, d\mu(x)
=\int_M \langle \nabla f(U_t(x)), \nabla\psi (U_t(x))\rangle\, d\mu(x),
\end{equation*}
which is equal to, for any $t\in [0,1]$, 
\begin{equation*}
\int_M \langle \nabla f, \nabla\psi\rangle\, dc_t.
\end{equation*}
In other term, $c_t$ is a solution to the following continuity equation: 
\begin{equation*}
\frac{dc_t}{dt}+ \nabla\cdot (\nabla\psi\,c_t)=0.
\end{equation*}

According to above definition, we see that for each $t\in [0,1]$, 
\begin{equation*}
\frac{{d}^I c_t}{dt}=\nabla\psi.
\end{equation*}

It is why we call $\nabla\psi$ a constant vector field on $\P_2(M)$. 
In order to make clearly different roles played by $\nabla\psi$, we will use notation  
 $$V_\psi$$
when it is seen as a constant vector field on $\P_2(M)$.

\begin{remark} {\rm In section \ref{sect3} below, we will compute Lie brackets of two constant vector fields on $\P_2(M)$ without explicitly using the existence of density of measure,  the Lie bracket of two constant vector fields is NOT a constant vector field.}
\end{remark}

\subsection{\bf Geodesics with constant speed}\label{section1.2}

It is easy to introduce geodesics with constant speed when the base space is a flat space $\R^m$. 
A probability measure $\mu$ on $\R^m$ is in $\P_2(\R^m)$ if $\int_{\R^m} |x|^2\, d\mu(x)<+\infty$. 
Let $c_0, c_1\in \P_2(\R^m)$, there is an optimal coupling plan $\gamma\in \C(c_0, c_1)$ such that 

\begin{equation*}
W_2^2(c_0, c_1)=\int_{\R^m\times\R^m} |x-y|^2\, d\gamma(x,y).
\end{equation*}

For each $t\in [0,1]$, define $c_t\in\P_2(\R^m)$ by
\begin{equation*}
\int_{\R^m} f(x)\, dc_t(x)=\int_{\R^m\times\R^m} f(u_t(x,y))\, d\gamma(x,y),
\end{equation*}
where $u_t(x,y)=(1-t)x+ty$. For $0\leq s<t\leq 1$, define $\pi_{s,t}\in \C(c_s,c_t)$ by 
\begin{equation*}
\int_{\R^m\times\R^m} g(x,y)\, d\pi_{s,t}(x,y)=\int_{\R^m\times\R^m} g(u_s(x,y), u_t(x,y))\, d\gamma(x,y).
\end{equation*}

Then
\begin{equation*}
W_2^2(c_s,c_t)\leq \int_{\R^m\times\R^m}|u_t(x,y)-u_s(x,y|^2\, d\gamma(x,y)
=(t-s)^2 W_2(c_0, c_1)^2.
\end{equation*}
It follows that $\dis W_2(c_s, c_t)\leq (t-s) W_2(c_0, c_1)$. Combing with triangulaire inequality, 
\begin{equation*}
\begin{split}
 W_2(c_0, c_1)&\leq W_2(c_0, c_s)+W_2(c_s,c_t)+W_2(c_t,c_1)\\
 &\leq s W_2(c_0, c_1) + (t-s)W_2(c_0,c_1)+ (1-t)W_2(c_0, c_1)=W_2(c_0, c_1),
 \end{split}
\end{equation*}
we get the property of geodesic with constant speed:

\begin{equation*}
W_2(c_s, c_t)=|t-s|\, W_2(c_0, c_1).
\end{equation*}

\vskip 2mm
According to Theorem \ref{th1}, there is $Z_t\in \TT_{c_t}$ such that, for $f\in C_c^1(\R^d)$, 

\begin{equation*}
\begin{split}
\frac{d}{dt} \int_{\R^m} f(x)dc_t(x)&=\int_{\R^m}\langle \nabla f(u_t(x,y)), y-x\rangle_{\R^m}\, d\gamma(x,y)\\
&=\int_{\R^d} \langle \nabla f(x), Z_t(x)\rangle_{\R^m}\, dc_t(x)
\end{split}
\end{equation*}
where $\langle\ ,\ \rangle_{\R^m}$ is the canonical inner product of $\R^m$. We heuristically look for $Z_t$ such that 
$$\dis Z_t(u_t(x,y))=y-x.$$

Taking the derivative with respect to $t$ yields 
\begin{equation*}
(\frac{d}{dt}Z_t)(u_t(x,y))+ \langle\nabla Z_t(u_t(x,y)), y-x\rangle = 0.
\end{equation*}
It follows that 
\begin{equation*}
(\frac{d}{dt}Z_t)+ \nabla Z_t(Z_t)=0.
\end{equation*}
In the case where $Z_t=\nabla \psi_t$, we have 
\begin{equation*}
(\frac{d}{dt}\nabla\psi_t)+ \nabla^2 \psi_t(\nabla\psi_t)=0.
\end{equation*}

We remark that $\{\nabla\psi_t, t\in ]0,1[\}$ satisfies heuristically the equation of Riemannian geodesic 
obtained in \cite{Lott} or heuristically obtained in \cite{OV}, in which the authors showed that the convexity of 
entropy functional along these geodesics is equivalent to Bakry-Emery's curvature condition \cite{BakryEmery}
(see also \cite{LiLi}, \cite{SV, Sturm}). 

\vskip 2mm
In the case of Riemannian manifold $M$, it is a bit complicated. We follow the exposition of \cite{Gigli}. Let $TM$ be the tangent
bundle of $M$ and $\pi : TM\ra M$ the natural projection. For each $\mu\in\P_2(M)$, we consider the set

\begin{equation*}
\Gamma_\mu=\Bigl\{ \gamma\ \hbox{\rm probability measure on } TM;\ \pi_\#\gamma=\mu,
\int_{TM} |v|_{T_xM}^2d\gamma(x,v)<+\infty \Bigr\}.
\end{equation*}
The set $\Gamma_\mu$ is obviously non empty. Let $\gamma\in \Gamma_\mu$, we consider $\nu=\exp_\#\gamma$,
that is,
\begin{equation*}
\int_M f(x)d\nu(x)=\int_{TM} f(\exp_x(v))\, d\gamma(x,v),
\end{equation*}
where $\exp_x: T_xM\ra M$ is the exponential map induced by geodesics on $M$. The map
\begin{equation*}
TM \ra M\times M,\quad (x,v)\ra (x, \exp_x(v))
\end{equation*}
sends $\gamma$ to a coupling plan $\tilde\gamma\in\C(\mu, \nu)$. We have
\begin{equation*}
W_2^2(\mu,\nu)\leq \int_{TM} d_M^2(x, \exp_x(v))\, d\gamma(x,v)\leq \int_{TM} |v|_{T_xM}^2\, d\gamma(x,v).
\end{equation*}
In order to construct geodesics $\{c_t; t\in [0,1]\}$ connecting $\mu$ and $\nu$, we need to find $\gamma_0\in \Gamma_\mu$
such that 
\begin{equation}\label{eq1.7}
W_2^2(\mu, \nu)=\int_{TM} |v|_{T_xM}^2\, d\gamma_0(x,v).
\end{equation}

\vskip 2mm
As $M$ is connected, let $x\in M$, for each $y$, there is a minimizing geodesic 
$\{\xi(t),\ t\in [0,1]\}$  connecting $x$ and $y$.
Let $v_{x,y}=\xi'(0)\in T_xM$, then 
\begin{equation*}
y=\exp_x(v_{x,y})\  \hbox{\rm and}\  d_M(x, y)=|v_{x,y}|_{T_xM}. 
\end{equation*}
Take a Borel version $\Xi$ of such a map $(x,y)\ra (x,v_{x,y})$ from $M\times M$ to $TM$. 
Let $\tilde\gamma_0\in \C(\mu,\nu)$ be an optimal coupling plan; define $\gamma_0\in\Gamma_\mu$ by

\begin{equation*}
\int_{TM} g(x,v)\,d\gamma_0(x,v)=\int_{M\times M} g\bigl(x,\Xi(x,y)\bigr)\, d\tilde\gamma_0(x,y).
\end{equation*}
Therefore
\begin{equation*}
\begin{split}
\int_{TM} |v|_{T_xM}^2\,d\gamma_0(x,v)&=\int_{M\times M} |\Xi(x,y)|^2\, d\tilde\gamma_0(x,y)\\
&=\int_{M\times M} d_M(x,y)^2\, d\tilde\gamma_0(x,y)=W_2^2(\mu, \nu).
\end{split}
\end{equation*}

\vskip 2mm
Now we define the curve $\{c_t;\ t\in [0,1]\}$ on $\P_2(M)$ by
\begin{equation*}
\int_M f(x)dc_t(x)=\int_{TM} f(\exp_x(tv))\, d\gamma_0(x,v).
\end{equation*}

Similarly we check that
\begin{equation*}
W_2(c_s, c_t)=|t-s|\, W_2(c_0, c_1).
\end{equation*}

\vskip 2mm
The organization of the paper is as follows. In Section 2, we consider ordinary equations on $\P_2(M)$, 
a Cauchy-Peano's type theorem is established, also Mckean-Vlasov equation involved. In Section 3, we emphasize 
that the suitable class of probability measures for developing  the differential geometry is one having divergence and the strictly positive density with certain regularity. The Levi-Civita connection is introduced and the formula for the covariant derivative 
of a general but smooth enough vector field is obtained. In section 4, we precise results on the derivability of the Wasserstein 
distance on $\P_2(M)$, which enable us to obtain the extension of a vector field along a quite good curve on $\P_2(M)$ in Section 5 as 
in differentiable geometry; the parallel translation along such a good curve on $\P_2(M)$ is naturally and rigorously introduced.
The existence for parallel translations is established for a curve whose intrinsic derivative gives rise a good enough 
vector field on $\P_2(M)$.

\section{Ordinary differential equations on $\P_2(M)$}\label{sect2}

Let $\varphi\in C^1(M)$, consider the function $F_\varphi$ on $\P_2(M)$ defined by
\begin{equation}\label{eq2.1}
F_\varphi(\mu)=\int_M \varphi(x)\, d\mu(x).
\end{equation}

A function $F$ on $\P_2(M)$ is said to be a polynomial if there exists a finite number of functions 
$\varphi_1, \ldots, \varphi_k$ in $C^1(M)$ such that $F=F_{\varphi_1}\cdots F_{\varphi_k}$.  Let $Z=V_\psi$ be a constant 
vector field on $\P_2(M)$ with $\psi\in C^{\infty}(M)$, and $U_t$ the flow on $M$ associated to $\nabla\psi$. For $\mu_0\in \P_2(M)$, we set 
$\mu_t=(U_t)_\#\mu_0$. Then we have seen in  section \ref{section1.1}, 
\begin{equation*}
\Bigl\{\frac{d}{dt} F_\varphi(\mu_t)\Bigr\}_{|_{t=0}}=\int_M \langle \nabla\varphi(x), \nabla\psi(x)\rangle\, d\mu_0(x)
=\langle V_\varphi, V_\psi\rangle_{\TT_{\mu_0}}.
\end{equation*}

The left hand side of above equality is the derivative of $F_\varphi$ along $V_\psi$. More generally, for a function $F$ on $\P_2(M)$,
we say that $F$ is derivable at $\mu_0$ along $V_\psi$, if
\begin{equation*}
(\bar{D}_{V_\psi}F)(\mu_0)=\Bigl\{\frac{d}{dt} F(\mu_t)\Bigr\}_{|_{t=0}}\quad\hbox{\rm exists}.
\end{equation*}

We say that the gradient $\bar{\nabla} F(\mu_0)\in \TT_{\mu_0}$ exists if for each $\psi\in C^\infty(M)$, $(\bar{D}_{V_\psi}F)(\mu_0)$
exists and
\begin{equation}\label{eq2.2}
\bar{D}_{V_\psi}F(\mu_0)=\langle \bar{\nabla}F, V_\psi\rangle_{\TT_{\mu_0}}.
\end{equation}
Note that for $\varphi\in C^1(M)$, there is a sequence of $\psi_n\in C^\infty(M)$ such that $\nabla\psi_n$ converge
uniformly to $\nabla\varphi$ so that $V_\varphi\in \TT_{\mu}$ for any $\mu\in \P_2(M)$. It is obvious that 
$\bar{\nabla}F_\varphi=V_\varphi$. For the
polynomial $F=\prod_{i=1}^kF_{\varphi_i}$, we have
\begin{equation*}
\bar{\nabla} F=\sum_{i=1}^k \Bigl(\prod_{j\neq i} F_{\varphi_j}\Bigr)\, V_{\varphi_i}.
\end{equation*}

\vskip 2mm
Note that the family $\{F_\varphi,\ \varphi\in C^1(M)\}$ separates the point of $\P_2(M)$. 
 By Stone-Weierstrauss theorem, the space of polynomials is dense in the space of continuous 
functions on $\P_2(M)$.
\vskip 2mm
{\bf Convention of notations:} We will use $\nabla$ to denote the gradient operator on the base space $M$,
and $\bar{\nabla}$ to denote the gradient operator on the Wasserstein space $(\P_2(M), W_2)$. For example, 
if $(\mu, x)\ra \Phi(\mu, x)$ is a function on $\P_2(M)\times M$, then $\nabla\Phi(\mu, x)$ is the gradient with respect to 
$x$, while $\bar{\nabla}\Phi(\mu,x)$ is the gradient with respect to $\mu$.

\begin{definition}
We will say that $Z$ is a vector field on $\P_2(M)$ if there exists a Borel map $\Phi: \P_2(M)\times M\ra\R$ 
such that for any $\mu\in\P_2(M)$, $\dis x\ra \Phi(\mu, x)$ is $C^1$ and
$\dis Z(\mu)=V_{\Phi(\mu, \cdot)}$. 
\end{definition}

\vskip 2mm
A  class of test vector fields on $\P_2(M)$ is 
\begin{equation}\label{eq2.2.1}
\chi(\P)=\Bigl\{ \sum_{finite} \alpha_i V_{\psi_i},\quad\alpha_i\ \hbox{\rm polynomial},\ \psi_i\in C^\infty(M)\Bigr\}.
\end{equation}

\vskip 2mm
Let $Z$ be a vector field on $\P_2(M)$, how to construct a solution $\mu_t\in\P_2(M)$ to the following ODE 
\begin{equation*}
\frac{{d}^I\mu_t}{dt}=Z(\mu_t) ?
\end{equation*}

\vskip 2mm
\begin{theorem}\label{th2.2} Let $Z$ be a vector field on $\P_2(M)$ given by $\Phi$. Assume that 
$\dis (\mu, x)\ra \nabla\Phi(\mu, x)$ is continuous, 
then for any $\mu_0\in\P_2(M)$, there is an absolutely curve $\{\mu_t;\ t\in [0,1]\}$ on $\P_2(M)$ such that 
\begin{equation}\label{eq2.3}
\frac{{d}^I\mu_t}{dt}=Z(\mu_t),\quad \mu_{|_{t=0}}=\mu_0.
\end{equation}
If moreover,   for any $\mu\in\P_2(M)$, $x\ra\Phi(\mu,x)$ is $C^2$ and 
\begin{equation}\label{eq2.4}
C_2:=\sup_{\mu\in\P_2(M)}\sup_{x\in M}||\nabla^2\Phi(\mu, x)|| <+\infty,
\end{equation}
then there is a flow of continuous maps $(t,x)\ra U_t(x)$ on $M$, solution to the following Mckean-Vlasov equation
\begin{equation}\label{eq2.5}
\frac{d}{dt}U_t(x) = \nabla\Phi(\mu_t, U_t(x)),\quad \mu_t=(U_t)_\#\mu_0.
\end{equation}
\end{theorem}

\vskip 2mm
\begin{proof}  We use the Euler approximation to construct a solution. We first note that
\begin{equation}\label{eq2.6}
C_1:=\sup_{(\mu,x)\in\P_2(M)\times M}|\nabla\Phi(\mu,x)|<+\infty.
\end{equation}
Let $\dis P_t=e^{t\Delta_M}$ be the heat semi-group associated to the Laplace operator $\Delta_M$ on functions, and
$\dis {\bf T_t}=e^{-t\square}$ the heat semigroup on differential forms, with de Rham-Hodge operator $\square$. It is 
well-known that

\begin{equation*}
|{\bf T}_t(\nabla\varphi)|\leq e^{-t\kappa/2} P_t|\nabla\varphi|, \quad \varphi\in C^1(M)
\end{equation*}
where $\kappa$ is lower bound of Ricci tensor on $M$. As $t\ra 0$, ${\bf T}_t(\nabla\varphi)$ converges to $\nabla\varphi$
uniformly. 
For $n\geq 1$, let 

\begin{equation*}
Z_n(\mu,x)=\bigl({\bf T_{1/n}}\nabla\Phi(\mu, \cdot)\bigr)(x).
\end{equation*}

According to \eqref{eq2.6} and above estimate, for $n$ big enough, 
\begin{equation}\label{eq2.7}
\sup_{(\mu,x)\in\P_2(M)\times M}|Z_n(\mu,x)|\leq 2C_1.
\end{equation}

Now let $t_k=k2^{-n}$ for $k=1,\ldots, 2^n$ and 
\begin{equation*}
[t]=t_k\quad\hbox{\rm if}\ t\in [t_k, t_{k+1}[.
\end{equation*}
On the intervall $[t_0, t_1]$, consider the ODE on $M$:

\begin{equation}\label{eq2.7.1}
\frac{dU_t^{(n)}}{dt}=Z_n\bigl(\mu_0, U_t^{(n)}\bigr),\quad U_0^{(n)}(x)=x,
\end{equation}
and $\dis \mu_t^{(n)}=(U_t^{(n)})_\#\mu_0$ for $t\in [t_0, t_1]$; 
inductively, on $\dis [t_k, t_{k+1}]$, we consider 
\begin{equation}\label{eq2.7.2}
\frac{dU_t^{(n)}}{dt}=Z_n\bigl(\mu_{t_k}^{(n)}, U_t^{(n)}\bigr),\quad U_{|_{t=t_k}}^{(n)}(x)=U_{t_k}^{(n)}(x),
\end{equation}
and for $t\in [t_k, t_{k+1}]$, 
\begin{equation}\label{eq2.7.3}
\mu_t^{(n)}=(U_t^{(n)})_\# \mu_{t_k}^{(n)}
\end{equation}

and so on, we get a curve $\{\mu_t^{(n)};\ t\in [0,1]\}$ on $\P_2(M)$. We now prove that this family is equicontinuous in 
$C([0,1], \P_2(M))$. 
Let $0\leq s<t\leq 1$, define $\gamma(\theta)=U_{(1-\theta)s+\theta t}^{(n)}$, then
\begin{equation*}
\frac{d\gamma(\theta)}{d\theta}
=(t-s) Z_n\bigl(\mu_{[(1-\theta)s+\theta t]}^{(n)}, U_{(1-\theta)s+\theta t}^{(n)}\bigr).
\end{equation*}
We have, according to \eqref{eq2.7},
\begin{equation*}
d_M\bigl(U_t^{(n)}(x), U_s^{(n)}(x)\bigr)\leq \int_0^1 \Bigl|\frac{d\gamma(\theta)}{d\theta}\Bigr|\, d\theta\leq 2C_1(t-s).
\end{equation*}
Define a probability measure $\pi$ on $M\times M$ by

\begin{equation*}
\int_{M\times M} g(x,y)\pi(dx,dy)=\int_M g\bigl(U_t^{(n)}(x), U_s^{(n)}(x)\bigr)\, d\mu_0(x).
\end{equation*}
Then $\pi\in \C(\mu_t^{(n)}, \mu_s^{(n)})$, we have
\begin{equation*}
W_2^2\bigl(\mu_t^{(n)}, \mu_s^{(n)}\bigr)\leq \int_M d_M^2\bigl(U_t^{(n)}(x), U_s^{(n)}(x)\bigr)\, d\mu_0(x)
\leq 4C_1^2\, (t-s)^2.
\end{equation*}

By Ascoli theorem, up to a subsequence, $\mu_\cdot^{(n)}$ converges in $C([0,1], \P_2(M))$ to a continuous curve 
$\{\mu_t;\ t\in [0,1]\}$ such that $\dis W_2(\mu_t, \mu_s)\leq 2C_1\, (t-s)$. 

\vskip 2mm
For proving that $\{\mu_t;\ t\in [0,1]\}$ is a solution to ODE \eqref{eq2.3}, we need the following preparation:
\begin{lemma} Set $\Phi_\mu(x)=\Phi(\mu,x)$, then
\begin{equation}\label{eq2.8}
\sup_{(\mu,x)\in\P_2(M)\times M}|({\bf T}_t\nabla\Phi_\mu)(x)-\nabla\Phi(x)|_{T_xM}\ra 0, \quad\hbox{\rm as}\ t\ra 0.
\end{equation}
\end{lemma}

\begin{proof} We use $||\cdot||_\infty$ to denote the uniform norm on $M$. Let $\eps>0$, for $\mu\in\P_2(M)$, 
there is $\hat t_\mu>0$ such that 
\begin{equation*}
\sup_{t\leq\hat t_\mu} ||{\bf T}_t\nabla\Phi_\mu-\nabla\Phi_\mu||_\infty  <\eps.
\end{equation*}
Since $\dis (\mu, t)\ra  ||{\bf T}_t\nabla\Phi_\mu-\nabla\Phi_\mu||_\infty$ is continuous, there is $\delta_\mu>0$ such that
for $t\leq \hat t_\mu$, 
\begin{equation*}
W_2(\mu,\nu)<\delta_\mu\ \Rightarrow\ 
||{\bf T}_t\nabla\Phi_\nu-\nabla\Phi_\nu||_\infty  <\eps.
\end{equation*}
Let $B(\mu, \delta)$ be the open ball in $(\P_2(M), W_2)$ centered at $\mu$, of radius $\delta$. We have
\begin{equation*}
\P_2(M)=\cup_{\mu\in \P_2(M)}B(\mu, \delta_\mu);
\end{equation*}
so there is a finite number of $\{\mu_1, \ldots, \mu_K\}$ such that 
\begin{equation*}
\P_2(M)=\cup_{i=1}^K B(\mu_i, \delta_{\mu_i}).
\end{equation*}
Let $\hat t=\min\bigl\{\hat t_{\mu_i},\  i=1, \ldots, K\bigr\}>0$.  Then for $0<t<\hat t$, 
\begin{equation*}
\sup_{\mu\in\P_2(M)}||{\bf T}_t\nabla\Phi_\mu-\nabla\Phi_\mu||_\infty\leq \eps.
\end{equation*}
So we get \eqref{eq2.8}.
\end{proof}

\vskip 2mm
{\bf End of the proof of theorem :}  $\{\mu_t^{(n}; \ t\in [0,1]\}$ satisfies the following continuity equation
\begin{equation}\label{eq2.9}
\begin{split}
&\int_{[0,1]\times M}\alpha'(t) f(x)d\mu_t^{(n)}(x)dt\\
&\hskip -8mm=\alpha(0)\,\int_M f(x)d\mu_0(x)
 + \int_{[0,1]\times M} \alpha(t)\, \langle \nabla f(x), Z_n\bigl( \mu_{[t]}^{(n)},x\bigr)\rangle\, d\mu_t^{(n)}(x)dt,
\end{split}
\end{equation}
for all $\alpha\in C_c^1([0,1))$ and $f\in C^1(M)$. We have

\begin{equation*}
\begin{split}
& \int_{[0,1]\times M} \alpha(t)\, \langle \nabla f(x), Z_n\bigl( \mu_{[t]}^{(n)},x\bigr)\rangle\, d\mu_t^{(n)}dt
 -  \int_{[0,1]\times M} \alpha(t)\, \langle \nabla f(x), \nabla\Phi\bigl( \mu_t,x\bigr)\rangle\, d\mu_tdt\\
 &= \int_{[0,1]\times M} \alpha(t)\, \langle \nabla f(x), Z_n\bigl( \mu_{[t]}^{(n)},x\bigr)-\nabla\Phi(\mu_t,x)\rangle\, d\mu_t^{(n)}dt\\
 &+ \int_{[0,1]\times M} \alpha(t)\, \langle \nabla f(x), \nabla\Phi\bigl( \mu_t,x\bigr)\rangle\, d\mu_t^{(n)}dt
 - \int_{[0,1]\times M} \alpha(t)\, \langle \nabla f(x), \nabla\Phi\bigl( \mu_t,x\bigr)\rangle\, d\mu_tdt.
\end{split}
\end{equation*}
It is obvious that the sum of two last terms converge to $0$ as $n\ra +\infty$. Let $I_n$ be the first term on the right side, then
\begin{equation*}
|I_n|\leq ||\nabla f||_\infty \int_0^1 |\alpha(t)|\, ||{\bf T}_{1/n}\nabla\Phi_{\mu_{[t]}^{(n)}}-\nabla\Phi_{\mu_t}||_\infty\, dt
\end{equation*}
Note that 
\begin{equation*}
||{\bf T}_{1/n}\nabla\Phi_{\mu_{[t]}^{(n)}}-\nabla\Phi_{\mu_t}||_\infty
\leq ||{\bf T}_{1/n}\nabla\Phi_{\mu_{[t]}^{(n)}}-\nabla\Phi_{\mu_{[t]}^{(n)}}||_\infty+
||\nabla\Phi_{\mu_{[t]}^{(n)}}-\nabla\Phi_{\mu_t}||_\infty.
\end{equation*}
The term $\dis ||{\bf T}_{1/n}\nabla\Phi_{\mu_{[t]}^{(n)}}-\nabla\Phi_{\mu_{[t]}^{(n)}}||_\infty\ra 0$ is due to above lemma. 
As $n\ra +\infty$, $\mu_{[t]}^{(n)}$ converges to $\mu_t$. By continuity of $(\mu, x)\ra \nabla\Phi(\mu, x)$, the last term
tends to $0$. Letting $n\ra +\infty$ in \eqref{eq2.9} yields 

\begin{equation*}
\begin{split}
&\int_{[0,1]\times M}\alpha'(t) f(x)d\mu_t(x)dt\\
&\hskip -8mm=\alpha(0)\,\int_M f(x)d\mu_0(x)
 + \int_{[0,1]\times M} \alpha(t)\, \langle \nabla f(x), \nabla\Phi\bigl( \mu_t,x\bigr)\rangle\, d\mu_t(x)dt,
\end{split}
\end{equation*}
which is the meaning of Equation \eqref{eq2.3} in distribution sense. 

\vskip 2mm
For the proof of second part, since $x\ra \Phi(\mu, x)$ is $C^2$, we can directly use $\nabla\Phi(\mu, \cdot)$ instead of 
$Z_n$ in \eqref{eq2.7.1}, \eqref{eq2.7.2}, \eqref{eq2.7.3}.

On the intervall $[t_0, t_1]$, consider the ODE on $M$:

\begin{equation}\label{eq2.13}
\frac{dU_t^{(n)}}{dt}=\nabla\Phi\bigl(\mu_0, U_t^{(n)}\bigr),\quad U_0^{(n)}(x)=x,
\end{equation}
and $\dis \mu_t^{(n)}=(U_t^{(n)})_\#\mu_0$ for $t\in [t_0, t_1]$; 
inductively, on $\dis [t_k, t_{k+1}]$, we consider 
\begin{equation}\label{eq2.14}
\frac{dU_t^{(n)}}{dt}=\nabla\Phi\bigl(\mu_{t_k}^{(n)}, U_t^{(n)}\bigr),\quad U_{|_{t=t_k}}^{(n)}(x)=U_{t_k}^{(n)}(x),
\end{equation}
and for $t\in [t_k, t_{k+1}]$, 
\begin{equation}\label{eq2.15}
\mu_t^{(n)}=(U_t^{(n)})_\# \mu_{t_k}^{(n)}.
\end{equation}

By above result, up to a subsequence, $\{\mu_t^{(n)},\ t\in [0,1]\}$ converges to $\{\mu_t, t\in [0,1]\}$ in $C([0,1], \P_2(M))$. 
We use this subsequence to prove the convergence of $\{U_t^{(n)}(x),\ t\in [0,1]\}$. Now we prove that, under Condition 
\eqref{eq2.6}, 

\begin{equation}\label{eq2.16}
d_M\Big( U_t ^{(n)}(x), U_t ^{(n)}(y)\Bigr) \leq e^{C_2 t}\, d_M(x,y),\quad x, y\in M.
\end{equation}

For $x,y \in M$ given, there is a minimizing geodesic $\{\xi_s,\ s\in [0,1]\}$ connecting $x$ and $y$ such that 
$ d_M(x,y)=\int_0^1|\xi_s'|\, ds$. Set 
\begin{equation*}
\sigma(t,s)=U_t^{(n)}(\xi_s).
\end{equation*}
Since the torsion is free, we have the relation: 
\begin{equation}\label{eq2.17}
\frac{D}{ds}\frac{d}{dt}\sigma(t,s)=\frac{D}{dt}\frac{d}{ds}\sigma(t,s),
\end{equation}
where $\frac{D}{ds}$ denotes the covariant derivative. We have
\begin{equation*}
\frac{d}{dt}U_t^{(n)}(\xi_s)=\nabla\Phi\Bigl(\mu_{[t]}^{(n)}, U_t^{(n)}(\xi_s)\Bigr).
\end{equation*}
Taking the derivative with respect to $s$, we get
\begin{equation*}
\frac{D}{ds}\frac{d}{dt}U_t^{(n)}(\xi_s)=\nabla^2\Phi\Bigl(\mu_{[t]}^{(n)}, U_t^{(n)}(\xi_s)\Bigr)\cdot \frac{d}{ds}U_t^{(n)}(\xi_s).
\end{equation*}
Combining with \eqref{eq2.17} yields 
\begin{equation*}
\frac{D}{dt}\frac{d}{ds}U_t^{(n)}(\xi_s)=\nabla^2\Phi\Bigl(\mu_{[t]}^{(n)}, U_t^{(n)}(\xi_s)\Bigr)\cdot \frac{d}{ds}U_t^{(n)}(\xi_s).
\end{equation*}
Now, 
\begin{equation*}
\frac{d}{dt}\Bigl|\frac{d}{ds}U_t^{(n)}(\xi_s)\Bigr|^2
=2\Bigl\langle \nabla^2\Phi\Bigl(\mu_{[t]}^{(n)}, U_t^{(n)}(\xi_s)\Bigr)\cdot \frac{d}{ds}U_t^{(n)}(\xi_s),\ \frac{d}{ds}U_t^{(n)}(\xi_s)\Bigr\rangle,
\end{equation*}
which is, by Condition \eqref{eq2.6}, less than
\begin{equation*}
 2C_2\, \Bigl|\frac{d}{ds}U_t^{(n)}(\xi_s)\Bigr|^2.
\end{equation*}

By Gronwall lemma,
\begin{equation*}
\Bigl| \frac{d}{ds}U_t^{(n)}(\xi_s)\Bigr|\leq e^{C_2t}\, |\xi_s'|,
\end{equation*}
which implies that 
\begin{equation*}
d_M\Bigl(U_t^{(n)}(x), U_t^{(n)}(y)\Bigr)\leq e^{C_2 t}\, d_M(x,y).
\end{equation*}
Therefore the family $\bigl\{ (t,x)\ra U_t^{(n)}(x);\ n\geq 1\bigr\}$ is equicontinuous in $C([0,1]\times M)$. By Ascoli theorem, 
up to a subsequence, $U_t^{(n)}(x)$ converges to $U_t(x)$ uniformly in $(t,x)\in [0,1]\times M$. It is obvious to see that 
$(U_t, \mu_t)$ solves Mckean-Vlasov equation \eqref{eq2.5}. 
\end{proof}

\begin{remark} {\rm Comparing to \cite{BLPR}, as well to \cite{Wang}, we did not suppose the Lipschitz continuity with respect to 
$\mu$; in counterpart, we have no uniqueness of solutions of \eqref{eq2.5}. }
\end{remark}

\begin{remark}{\rm Many interesting PDE can be interpreted as gradient flows on the Wasserstein space $\P_2(M)$ (see \cite{AGS}, 
\cite{Villani1},\cite{Villani2}, \cite{FangShao}). The interpolation between geodesic flows and gradient flows were realized using 
Langevin's deformation in \cite{LiLi, LiLi2}.
}
\end{remark}

\section{Levi-Civita connection on $\P_2(M)$}\label{sect3}

In this section, we will revisit the paper by J. Lott \cite{Lott}: we try to reformulate conditions given there as weak as possible, 
also to expose some of them in an intrinsic way, avoiding the use of density. In order to obtain good pictures on the geometry of $\P_2(M)$,
the suitable class of probability measures should be the class $\dis\P_{\div}(M)$ of probability measures on $M$
  having divergence (see Definition \ref{def1.1}). 
 
 \vskip 2mm
 For convenience of readers, we will briefly prepare materials needed for our exposition.
  For a measure $\mu\in \P_2(M)$,  for any $C^1$ vector field $A$ on $M$, 
  the divergence $\div_\mu(A)\in L^2(M,\mu)$ is such that 
 
 \begin{equation*}
 \int_M \langle \nabla\phi(x), A(x)\rangle_{T_xM}\, d\mu(x)=-\int_M \phi(x)\,\div_\mu(A)(x)\, d\mu(x)
 \end{equation*}
   for any $\phi\in C^1(M)$.
 It is easy to see that $\div_\mu(fA)=f\,\div_\mu(A)+ \langle \nabla f, A\rangle$ for $f\in C^1(M)$.
  If $d\mu=\rho\, dx$ has a density $\rho>0$ in the space $C^1(M)$, we have 
 \begin{equation*}
 \int_M \langle\nabla\phi, A\rangle\, d\mu=\int_M \langle\nabla\phi, \rho A\rangle\, dx=-\int_M \phi\,\div(\rho A)\, dx
 =-\int_M \phi\,\div(\rho A)\,\rho^{-1} d\mu,
 \end{equation*}
It follows that
 \begin{equation}\label{eq3.01}
 \div_\mu(A)=\rho^{-1}\,\div(\rho A)=\div(A)+\langle \nabla(\log \rho), A\rangle.
 \end{equation}
 
 For $\mu\in \P_{\div}(M)$ and $\phi\in C^2(M)$, we denote $\L^\mu(\phi)\in L^2(\mu)$ such that
 \begin{equation}\label{eq3.1}
 \int_M \langle \nabla f, \nabla\phi\rangle\, d\mu=-\int_M f\,\L^\mu\phi\, d\mu,\quad\hbox{\rm for any } f\in C^1(M),
 \end{equation}
where $\dis \L^\mu\phi=\div_\mu(\nabla\phi)$ is a negative operator.

\vskip 2mm
Let $\psi\in C^3(M)$, consider the ODE
 \begin{equation*}
\frac{dU_t}{dt}=\nabla\psi(U_t),\quad U_0(x)=x.
 \end{equation*}
 
\begin{proposition} Let $d\mu=\rho\, dx$ be a probability measure in $\P_{\div}(M)$ 
with a strictly positive density $\rho$ in $C^1(M)$ and $\psi\in C^3(M)$. 
Then for each $t\in [0,1]$, $\mu_t:=(U_t)_\#\mu\in \P_{div}(M)$. 
\end{proposition}

\begin{proof} By Kunita \cite{Kunita} (see also \cite{ABC}, \cite{Malliavin}), the push-forward measure $(U_t^{-1})_\#\mu$ 
by inverse map of $U_t$ admits a density
$\tilde K_t$ with respect to $\mu$, having the following explicit expression
\begin{equation*}
\tilde K_t=\exp\Bigl(-\int_0^t \div_{\mu}(\nabla\psi)(U_s(x))ds\Bigr).
\end{equation*}
It follows that the density $K_t$ of $\mu_t$ with respect to $\mu$ has the expression 
\begin{equation*}
 K_t=\exp\Bigl(\int_0^t \div_{\mu}(\nabla\psi)(U_{-s}(x))ds\Bigr).
\end{equation*}

According to \eqref{eq3.01}, $\dis x\ra \div_\mu(\nabla\psi(x))$ is $C^1$.
Therefore the condition in \cite{ABC}
\begin{equation*}
\int_M \exp(\lambda \div_\mu(\nabla\psi(x))\, d\mu(x)<+\infty, \ \hbox{\rm for all } \lambda>0
\end{equation*}
is automatically satisfied.  Again by \eqref{eq3.01}, $x\ra K_t(x)$ is in $C^1$. Now let $A$ be a $C^1$ vector field 
on $M$ and $f\in C^1(M)$, we have
\begin{equation*}
\int_M \langle \nabla f(x), A(x)\rangle_{T_xM}\, d\mu_t(x)
=\int_M \langle \nabla f, A\rangle_{T_xM}\, K_t(x)d\mu(x)
=-\int_M f\,\div_{\mu}(K_t Z)\, d\mu.
\end{equation*}
It follows that 
\begin{equation*}
\div_{\mu_t}(A)=\div_{\mu}(K_t A)\, K_t^{-1}.
\end{equation*}
\end{proof}

\vskip 2mm
For $\psi_1, \psi_2\in C^2(M)$, we denote by $V_{\psi_1}, V_{\psi_2}$ the associated constant vector fields on $\P_2(M)$.
In what follows, we will compute the Lie bracket $[V_{\psi_1}, V_{\psi_2}]$.
\vskip 2mm
For $f\in C^1(M)$, we set $F_f(\mu)=\int_M f\, d\mu$. According to preparations given at the beginning of  Section \ref{sect2}, 
\begin{equation*}
(\bar{D}_{V_{\psi_2}}F_f)(\mu)=\int_M \langle\nabla\psi_2, \nabla f\rangle\, d\mu=F_{\langle\nabla\psi_2,\nabla f\rangle}(\mu).
\end{equation*}

Using again above formula, we have
\begin{equation*}
(\bar{D}_{V_{\psi_1}}\bar{D}_{V_{\psi_2}}F_f)(\mu)=\int_M\langle \nabla\psi_1, \nabla\langle \nabla\psi_2,\nabla f\rangle\rangle\,d\mu
=-\int_M \L^\mu\psi_1\,\, \langle \nabla\psi_2,\nabla f\rangle\,d\mu.
\end{equation*}
Therefore 
\begin{equation*}
\begin{split}
[V_{\psi_2}, V_{\psi_1}] F_f&=\bar{D}_{V_{\psi_2}}\bar{D}_{V_{\psi_1}}F_f - \bar{D}_{V_{\psi_1}}\bar{D}_{V_{\psi_2}}F_f\\
&=\int_M \langle (\L^\mu\psi_1\,\nabla\psi_2-\L^\mu\psi_2\, \nabla\psi_1),\ \nabla f\rangle\, d\mu.
\end{split}
\end{equation*}
Let 
\begin{equation}\label{eq3.2}
\C_{\psi_1,\psi_2}(\mu)= \L^\mu\psi_1\,\nabla\psi_2-\L^\mu\psi_2\, \nabla\psi_1.
\end{equation}
Note that $\C_{\psi_1,\psi_2}(\mu)$ is in $L^2(M,TM; \mu)$, not in $\TT_\mu$. Consider the orthogonal projection: 
\begin{equation*}
\Pi_\mu: L^2(M, TM; \mu)\ra \TT_\mu. 
\end{equation*}

 As $\mu\in \P_{div}(M)$ and by Proposition \ref{prop1.1},  there exists $\tilde\Phi_\mu\in\D_1^2(\mu)$ such that 
\begin{equation}\label{eq3.3}
\Pi_\mu(\C_{\psi_1,\psi_2}(\mu))=\nabla \tilde\Phi_\mu.
\end{equation}
Then we have
\begin{equation}\label{eq3.4}
[V_{\psi_2}, V_{\psi_1}] F_f=\int_M \langle \nabla\tilde\Phi_\mu, \ \nabla f\rangle\, d\mu
=(\bar{D}_{V_{\tilde\Phi_\mu}}F_f)(\mu).
\end{equation}

Above equality can be extended to the class of polynomials on $\P_2(M)$, that is to say that 
\begin{equation}\label{eq3.5}
[V_{\psi_2}, V_{\psi_1}] _\mu= V_{\tilde\Phi_\mu}\quad \hbox{\rm on polynomials},
\end{equation}
We emphasize that Lie bracket of two constant vector fields is no more a constant vector field. 

\begin{proposition}\label{prop3.2}  Let $\psi_1, \psi_2\in C^3(M)$, for $d\mu=\rho\, dx$ with $\rho>0$ and $\rho\in C^2(M)$, 
the function $\tilde\Phi_\mu$ obtained in \eqref{eq3.3}
has the following expression :

\begin{equation}\label{eq3.6}
\tilde\Phi_\mu=(\L^\mu)^{-1}\ \div_\mu\bigl( \C_{\psi_1,\psi_2}(\mu)\bigr).
\end{equation}
\end{proposition}

\begin{proof}   By \eqref{eq3.01}, 
\begin{equation*}
\L^\mu\psi= \Delta_M \psi + \langle\nabla\log\rho, \nabla\psi\rangle,
\end{equation*}
where $\Delta_M$ denotes the Laplace operator on $M$. It is well-known that $\L^\mu$ has a spectral gap
 if $\log\rho\in C^2(M)$. In \cite{Lott}, the Lie bracket $[V_{\psi_2}, V_{\psi_1}]$ was expressed 
 using Hodge decomposition for  vector fields in $L^2(\mu)$.
 For $\psi_1, \psi_2\in C^3(M)$,  we have
\begin{equation*}
\div_\mu\bigl( \C_{\psi_1,\psi_2}(\mu)\bigr) =\langle \nabla \L^\mu\psi_1,\ \nabla\psi_2\rangle -\langle\nabla\L^\mu\psi_2,\ \nabla\psi_1\rangle.
\end{equation*}

By Hodge decomposition,  $\C_{\psi_1,\psi_2}(\mu)$ admits the decomposition
\begin{equation*}
\C_{\psi_1,\psi_2}(\mu)={d_\mu}^*\omega + \nabla f + h,
\end{equation*}
where $\omega$ is a differential $2$-form on $M$, ${d_\mu}^*$ is adjoint operator of exterior derivative in $L^2(\mu)$, 
$h$ is harmonic form : $\dis ({d_\mu}^*d+d{d_\mu}^*)h=0$. Taking the divergence $\div_\mu$ on the two sides of above 
equality, we see that $f$ is a solution the following equation
\begin{equation*}
\L^\mu f=\div_\mu\bigl( \C_{\psi_1,\psi_2}(\mu)\bigr).
\end{equation*}
It follows that $\tilde\Phi_\mu$ has the expression \eqref{eq3.6}.
\end{proof}

\vskip 2mm
Now we introduce the covariant derivative $\bn_{V_{\psi_1}}V_{\psi_2}$ associated to the Levi-Civita 
connection on $\P_2(M)$ by

\begin{equation*}
\begin{split}
2\langle \bn_{V_{\psi_1}}V_{\psi_2}, V_{\psi_3}\rangle
&= \bD_{V_{\psi_1}}\langle V_{\psi_2}, V_{\psi_3}\rangle + \bD_{V_{\psi_2}}\langle V_{\psi_3}, V_{\psi_1}\rangle 
- \bD_{V_{\psi_3}}\langle V_{\psi_1}, V_{\psi_2}\rangle \\
&+ \langle V_{\psi_3}, [V_{\psi_1},V_{\psi_2}]\rangle
- \langle V_{\psi_2}, [V_{\psi_1},V_{\psi_3}]\rangle - \langle V_{\psi_1}, [V_{\psi_2},V_{\psi_3}]\rangle.
\end{split}
\end{equation*}

We have $\dis \langle V_{\psi_2}, V_{\psi_3}\rangle =\int_M \langle\nabla\psi_2, \nabla\psi_3\rangle\, d\mu
=F_{ \langle\nabla\psi_2, \nabla\psi_3\rangle}$. Then
\begin{equation*}
 \bD_{V_{\psi_1}}\langle V_{\psi_2}, V_{\psi_3}\rangle
 =\int_M \langle\nabla\psi_1, \nabla\ \langle\nabla\psi_2, \nabla\psi_3\rangle\rangle\, d\mu
 =-\int_M \langle \L^\mu\psi_1\,\nabla\psi_2,\ \nabla\psi_3\rangle\, d\mu.
 \end{equation*}
 
Replacing $\psi_1$ by $\psi_2$, $\psi_2$ by $\psi_3$ and $\psi_3$ by $\psi_1$, we get
\begin{equation*}
 \bD_{V_{\psi_2}}\langle V_{\psi_3}, V_{\psi_1}\rangle
 =-\int_M \langle \L^\mu\psi_2\,\nabla\psi_1,\  \nabla\psi_3\rangle\, d\mu.
 \end{equation*}
 We have, in the same way
 
\begin{equation*}
 \bD_{V_{\psi_3}}\langle V_{\psi_1}, V_{\psi_2}\rangle
 =-\int_M \langle \L^\mu\psi_3\,\nabla\psi_1,\  \nabla\psi_2\rangle\, d\mu.
 \end{equation*}
 
 Now using expression of $ [V_{\psi_1},V_{\psi_2}]$, we have
 \begin{equation*}
 \langle V_{\psi_3}, [V_{\psi_1},V_{\psi_2}]\rangle
 =\int_M \langle -\L^\mu\psi_1\,\nabla\psi_2+\L^\mu\psi_2\, \nabla\psi_1, \nabla\psi_3\rangle\ d\mu.
 \end{equation*}
In the same way, we get 
\begin{equation*}
 \langle V_{\psi_2}, [V_{\psi_1},V_{\psi_3}]\rangle
 =\int_M \langle -\L^\mu\psi_1\,\nabla\psi_3+\L^\mu\psi_3\, \nabla\psi_1, \nabla\psi_2\rangle\ d\mu
 \end{equation*}
and
\begin{equation*}
 \langle V_{\psi_1}, [V_{\psi_2},V_{\psi_3}]\rangle
 =\int_M \langle -\L^\mu\psi_2\,\nabla\psi_3+\L^\mu\psi_3\, \nabla\psi_2, \nabla\psi_1\rangle\ d\mu.
 \end{equation*}
Combining all these terms, we  finally get
\begin{equation*}
2\langle \bn_{V_{\psi_1}}V_{\psi_2}, V_{\psi_3}\rangle
=\int_M \langle \nabla\langle\nabla\psi_1, \nabla\psi_2\rangle,\ \nabla\psi_3\rangle\, d\mu
+\int_M \langle \L^\mu\psi_2\,\nabla\psi_1-\L^\mu\psi_1\, \nabla\psi_2,\ \nabla\psi_3\rangle\, d\mu.
\end{equation*}

\begin{theorem}\label{th3.1}
For two constant vector fields $V_{\psi_1}, V_{\psi_2}$, we have
\begin{equation}\label{eq3.7}
\bn_{V_{\psi_1}}V_{\psi_2}=\frac{1}{2} V_{\langle\nabla\psi_1, \nabla\psi_2\rangle} + \frac{1}{2}[V_{\psi_1}, V_{\psi_2}].
\end{equation}
Moreover, for any constant vector field $V_{\psi_3}$, 
\begin{equation}\label{eq3.7-1}
\langle \bn_{V_{\psi_1}}V_{\psi_2},\ V_{\psi_3}\rangle_{\TT_\mu} 
= \int_M \langle \nabla^2\psi_2,\ \nabla\psi_1\otimes \nabla \psi_3\rangle\  d\mu.
\end{equation}
\end{theorem}

\begin{proof} It is enough to prove \eqref{eq3.7-1}. We have 
\begin{equation*}
\begin{split}
 \langle V_{\psi_3}, [V_{\psi_1},V_{\psi_2}]\rangle_{\TT_\mu}
 &=\int_M \langle -\L^\mu\psi_1\,\nabla\psi_2+\L^\mu\psi_2\, \nabla\psi_1, \nabla\psi_3\rangle\ d\mu\\
 &=\int_M \langle \nabla\psi_1,\ \nabla\langle\nabla\psi_2,\nabla\psi_3\rangle\rangle\, d\mu
 -\int_M \langle \nabla\psi_2,\ \nabla\langle\nabla\psi_1,\nabla\psi_3\rangle\rangle\, d\mu\\
 &=\int_M \Bigl( \langle \nabla^2\psi_2,  \nabla\psi_1\otimes\nabla\psi_3\rangle 
 + \langle \nabla^2\psi_3,  \nabla\psi_1\otimes\nabla\psi_2\rangle\Bigr)  d\mu\\
 &\hskip 4mm
 - \int_M \Bigl( \langle \nabla^2\psi_1,  \nabla\psi_2\otimes\nabla\psi_3\rangle 
 + \langle \nabla^2\psi_3,  \nabla\psi_2\otimes\nabla\psi_1\rangle\Bigr)  d\mu\\
 &= \int_M \Bigl( \langle \nabla^2\psi_2,  \nabla\psi_1\otimes\nabla\psi_3\rangle 
 - \langle \nabla^2\psi_1,  \nabla\psi_2\otimes\nabla\psi_3\rangle \rangle\Bigr)  d\mu, 
 \end{split}
 \end{equation*}
due to the symmetry of the Hessian $\nabla^2\psi_3$. On the other hand, 

\begin{equation*} 
 \langle V_{\psi_3}, V_{\langle\nabla\psi_1, \nabla\psi_2\rangle}\rangle_{\TT_\mu}
 = \int_M \Bigl( \langle \nabla^2\psi_2,  \nabla\psi_3\otimes\nabla\psi_1\rangle 
 + \langle \nabla^2\psi_1,  \nabla\psi_3\otimes\nabla\psi_2\rangle \rangle\Bigr)  d\mu.
\end{equation*}
Summing these last two equalities yields \eqref{eq3.7-1}. 
\end{proof}

\begin{remark} {\rm By \eqref{eq3.7}, for two constant vector fields $V_{\psi_1}, V_{\psi_2}$, 
the covariant derivative $\bar{\nabla}_{V_{\psi_1}}V_{\psi_2}$ is not a constant vector field on $\P_2(M)$ if 
$\psi_1\neq \psi_2$. 
}
\end{remark}

\vskip 2mm
Let $\alpha: \P_2(M)\ra\R$ be a differentiable function,  we define
\begin{equation}\label{eq3.10}
\bn_{V_{\psi_1}}\bigl(\alpha\, V_{\psi_2}\bigr)=\bD_{V_{\psi_1}}\alpha\cdot V_{\psi_2} + \alpha\, \bn_{V_{\psi_1}} V_{\psi_2}.
\end{equation}

\begin{proposition} Let $Z$ be a vector field on $\P_2(M)$ in the test space $\chi(\P)$, that is, 
$\dis Z=\sum_{i=1}^k \alpha_i\, V_{\psi_i}$ with $\alpha_i$ polynomials. Then $\bn_{Z}Z$  still is in the test space;
 moreover 
\begin{equation*}
\bn_ZZ=V_{\Phi_1} + \frac{1}{2} V_{|\nabla\Phi_2|^ 2},
\end{equation*}
where
\begin{equation*}
\Phi_1=\sum_{j=1}^k \Bigl(\sum_{i=1}^k\alpha_i\,\bD_{V_{\psi_i}}\alpha_j\Bigr)\ \psi_j,\quad 
\Phi_2=\sum_{i=1}^k \alpha_i\,\psi_i.
\end{equation*}
\end{proposition}

\begin{proof} Using the rule concerning covariant derivatives, $\bn_ZZ$ is equal to
\begin{equation*}
\sum_{i,j=1}^k \alpha_i\,\bigl(\bD_{V_{\psi_i}}\alpha_j\bigr)\ V_{\psi_j} 
+ \frac{1}{2}\sum_{i,j=1}^k \alpha_i\alpha_j V_{\langle\nabla\psi_i, \nabla\psi_j\rangle}
+\frac{1}{2} \sum_{i,j=1}^k \alpha_i\alpha_j [V_{\psi_i}, V_{\psi_j}]. 
\end{equation*}
The last sum is equal to $0$ due to the skew-symmetry of $ [V_{\psi_i}, V_{\psi_j}]$, the first one gives rise to $\Phi_1$ and
the second one gives rise to $\Phi_2$.
\end{proof}

In what follows, we will extend the definition of covariant derivative \eqref{eq3.10} for a general vector field $Z$ on 
$\P_2(M)$. Let $\Delta$ be the Laplace operator on $M$, let $\{\varphi_n,\ n\geq 0\}$ be the eigenfunctions of $\Delta$:

\begin{equation*}
-\Delta \varphi_n=\lambda_n\, \varphi_n.
\end{equation*}

We have $\lambda_0=0$ and $\varphi_0=1$. It is well-known, by Weyl's result, that 
\begin{equation*}
\lambda_n\sim n^{2/m},\quad n\ra +\infty
\end{equation*}
where $m$ is the dimension of $M$. The functions $\{\varphi_n;\ n\in \N\}$ are smooth, chosen to form an orthonormal basis of 
$\dis L^2(M, dx)$. A function $f$ on $M$ is said to be in $H^k(M)$ for $k\in\N$, if
\begin{equation*}
||f||_{H^k}^2=\int_M |(I-\Delta)^{k/2}f|^2\, dx< +\infty.
\end{equation*}

By Sobolev embedding inequality, for $\dis k>\frac{m}{2}+ q$, 
\begin{equation*}
||f||_{C^q}\leq C\, ||f||_{H^k}.
\end{equation*}

For $f\in H^k(M)$, put $\dis f=\sum_{n\geq 0} a_n \varphi_n$ which holds in $L^2(M, dx)$ with 
\begin{equation*}
a_n=\int_M f(x)\varphi_n(x)\,dx.
\end{equation*}
We have : 
\begin{equation*}
||f||_{H^k}^2 =\sum_{n\geq 0} a_n^2\, (1+\lambda_n)^k.
\end{equation*}
The system 
$\dis\Bigl\{\frac{\nabla\varphi_n}{\sqrt{\lambda_n}};\quad n\geq 1 \Bigr\}$
is orthonormal. Let $\dis V_n =V_{\varphi_n/\sqrt{\lambda_n}}$, then $\{V_n;\ n\geq 1\}$ is an orthonormal basis of $\TT_{dx}$.

\vskip 2mm
Let $Z$ be a vector field on $\P_2(M)$ given by $Z(\mu)=V_{\Phi(\mu, \cdot)}$ or $Z(\mu)=\nabla\Phi(\mu, \cdot)$. 
In the sequel, we denote: $\dis \Phi_\mu(x)=\Phi(\mu, x)$, $\dis \Phi^x(\mu)=\Phi(\mu, x)$. Then, if $x\ra \nabla\Phi_\mu(x)$
is continuous, 
\begin{equation*}
\nabla\Phi_\mu = \sum_{n\geq 1} \Bigl(\int_M \langle \nabla\Phi_\mu, \frac{\nabla\varphi_n}{\sqrt{\lambda_n}}\rangle\, dx\Bigr)\,
\frac{\nabla\varphi_n}{\sqrt{\lambda_n}}
= \sum_{n\geq 1} \Bigl(\int_M \Phi_\mu\varphi_n dx\Bigr)\,{\nabla\varphi_n},
\end{equation*}
which converges in $L^2(M, dx)$. Let $\mu\in\P_{\div}(M)$, the above series converges also in $\TT_\mu$. Let

\begin{equation}\label{eq3.11}
a_n(\mu)=\int_M \Phi_\mu(x)\varphi_n(x)\,dx.
\end{equation}

Let $V_\psi$ be a constant vector field on $\P_2(M)$ with $\psi\in C^\infty(M)$. For $q\geq p\geq 1$, set
\begin{equation}\label{eq3.12}
S_{p,q}=\sum_{n=p}^q \Bigl( \bD_{V_\psi}a_n\,V_{\varphi_n} + a_n\, \bn_{V_\psi}V_{\varphi_n}\Bigr)=S_{p,q}^1+S_{p,q}^2
\end{equation}
respectively. Let $\phi\in C^\infty(M)$, according to \eqref{eq3.7-1}, we have
\begin{equation*}
\langle S_{p,q}^2, V_\phi\rangle_{\TT_\mu}
=\int_M \Bigl( \sum_{n=p}^q a_n(\mu)\nabla^2\varphi_n\Bigr)(\nabla\psi(x), \nabla\phi(x))\, d\mu(x).
\end{equation*}
It follows that 
\begin{equation*}
|\langle S_{p,q}^2, V_\phi\rangle_{\TT_\mu}|
\leq  \Bigl\| \sum_{n=p}^q a_n(\mu)\nabla^2\varphi_n\Bigr\|_{\infty}\, |V_{\psi}|_{\TT_\mu} |V_{\phi}|_{\TT_\mu},
\end{equation*}
therefore 
\begin{equation*}
|S_{p,q}^2|_{\TT_{\mu}}
\leq  \Bigl\| \sum_{n=p}^q a_n(\mu)\nabla^2\varphi_n\Bigr\|_{\infty}\, |V_{\psi}|_{\TT_\mu}.
\end{equation*}

We have
\begin{equation*}
\begin{split}
 & ||\sum_{n=p}^q a_n(\mu) (I-\Delta)^{k/2}\varphi_n||_{L^2(dx)}^2
  =\sum_{n=p}^q a_n(\mu)^2(1+\lambda_n)^k\\
  &=\sum_{n=p}^q \Bigl(\int_M(I-\Delta)^{k/2}\Phi_\mu\,\varphi_n\,dx\Bigr)^2\ra 0
  \end{split}
\end{equation*}
as $p, q\ra +\infty$ if $\Phi_\mu\in H^k(M)$. On the other hand, we have
\begin{equation*}
(\bD_{V_\psi}a_n)(\mu)=\int_M (\bD_{V_\psi}\Phi^x)(\mu)\varphi_n(x)\,dx
=\int_M \langle\nabla\bD_{V_\psi}\Phi^x,\frac{\nabla\varphi_n}{\sqrt{\lambda_n}}\rangle\,\frac{dx}{\sqrt{\lambda_n}},
\end{equation*}
then
\begin{equation*}
S_{p,q}^1
=  \sum_{n=p}^q\Bigl(\int_M \langle\nabla\bD_{V_\psi}\Phi^x,\frac{\nabla\varphi_n}{\sqrt{\lambda_n}}\rangle\, dx\Bigr)
\frac{\nabla\varphi_n}{\sqrt{\lambda_n}}
\end{equation*}
and 
\begin{equation*}
\int_M |S_{p,q}^1|^2\, dx
=  \sum_{n=p}^q\Bigl(\int_M \langle\nabla\bD_{V_\psi}\Phi^x,\frac{\nabla\varphi_n}{\sqrt{\lambda_n}}\rangle\, dx\Bigr)^2
\ra 0
\end{equation*}
as $p, q\ra +\infty$ if 
\begin{equation*}
\int_M |\nabla \bD_{V_\psi}\Phi^x|^2\, dx <+\infty.
\end{equation*}

Therefore for $d\mu=\rho\,dx$ with $\mu\in \P_{\div}(M)$, as $p,q\ra \infty$,
\begin{equation*}
|S_{p,q}^1|_{\TT_\mu}^2\leq ||\rho||_{\infty} \int_M |S_{p,q}^1|^2\, dx\ra 0.
\end{equation*}

We get the following result 

\begin{theorem}\label{th3.6} Let $Z$ be a vector field on $\P_2(M)$ given by $\Phi : \P_2(M)\times M\ra \R$.
Assume that 
\vskip 2mm
\quad (i) for any $\mu\in \P_2(M)$, $\Phi_\mu\in H^k(M)$ with $\dis k>\frac{m}{2}+2$,
\vskip 1mm

\quad (ii) for any $x\in M, \bD_{V_\psi}\Phi^x$ exists and $\dis\nabla \bD_{V_\psi}\Phi^{\cdot}\in L^2(M, dx)$.
\vskip 2mm

Then the covariant derivative $\bn_{V_\psi}Z$ is well defined at $\mu\in\P_{\div}(M)$ and for $\phi\in C^\infty(M)$,
\begin{equation}\label{eq3.13}
\langle \bn_{V_\psi}Z, V_\phi\rangle_{\TT_\mu}
=\int_M\langle (\nabla \bD_{V_\psi}\Phi^{\bf\cdot}), \nabla\phi\rangle\, d\mu
+\int_M \nabla^2\Phi_\mu\bigl(\nabla\psi, \nabla\phi\bigr)\, d\mu.
\end{equation}
\end{theorem}

\begin{proof} Let $\dis Z_q=\sum_{n=1}^q  a_n V_{\varphi_n}$. Then 
\begin{equation*}
\bn_{V_\psi} Z_q= S_{1,q}.
\end{equation*}
Letting $q\ra +\infty$ yields the result.
\end{proof}

\section{Derivability of the square of the Wasserstein distance}\label{sect4}

 Let $\{c_t;\ t\in [0,1]\}$ be an absolutely continuous curve on $\P_2(M)$, for $\sigma\in\P_2(M)$ given, 
  the derivability of $\dis t\ra W_2^2(\sigma, c_t)$
was established in chapter 8 of \cite{AGS} , as well as in \cite{Villani1} (see pages 636-649); however they 
hold true only for almost all $t\in [0,1]$. The derivability at $t=0$ was proved in Theorem 8.13 of \cite{Villani2} if 
$\sigma$ and $c_0$ have a density with respect to $dx$. When $\{c_t\}$ is a geodesic of constant speed, 
the derivability at $t=0$ was given in  theorem 4.2 of \cite{Gigli} where the property of semi concavity was used.
In what follows, we will use constant vector fields on $\P_2(M)$. 

\vskip 2mm
Before stating our result, we recall some well-known 
facts concerning optimal transport maps (see \cite{BB, Brenier, McCann, AGS, Villani1}).
Let $\sigma\in\P_{2, ac}(M)$ be absolutely continuous with respect to $dx$ and $\mu\in \P_2(M)$, then there is 
an unique Borel map $\phi \in \D_1^2(\sigma)$  such that 
\begin{equation*}
\int_M |\nabla\phi(x)|^2\, d\sigma(x)=W_2^2(\sigma,\mu)
\end{equation*}
and $x\ra T(x)=\exp_x(\nabla\phi(x))$ pushes $\sigma$ forward to $\mu$. If $\mu$ is also in $\P_{2, ac}(M)$, 
the map $T: M\ra M$ is invertible  and its inverse map $T^{-1}$ is given by $y\ra \exp_y(\nabla\tilde{\phi}(y))$ with 
some function $\tilde\phi$ such that $\int_M |\nabla\tilde{\phi}|^2d\mu<+\infty$. We need also the following result 

\begin{lemma}\label{lemma4.1} Let $x,y\in M$ and $\{\xi(t);\ t\in [0,1]\}$ be a minimizing geodesic 
connecting $x$ and $y$, given by $\dis \xi(t)=\exp_x(tu)$ with some $u\in T_xM$. Then 
\begin{equation}\label{eq4.1}
d_M^2(\exp_y (v), x)-d_M^2(y,x)\leq 2 \langle v, \xi'(1)\rangle_{T_yM}+ o(|v|)\quad\hbox{\rm as}\ |v|\ra 0.
\end{equation}
\end{lemma}
\begin{proof} See \cite{McCann}, page 10.
\end{proof}

\begin{theorem}\label{th4.2}
 Assume that $\sigma\in \P_{2, ac}(M)$ is absolutely continuous with respect to $dx$, then 
 $\dis \mu\ra \chi(\mu):=W_2^2(\sigma, \mu)$ is derivable along each constant vector field $V_\psi$ at any $\mu\in\P_2(M)$.
 If $\mu\in\P_{2,ac}(M)$, the gradient $\nabla\chi$ exists and admits the expression :
 \begin{equation}\label{eq4.2}
\nabla\chi(\mu) = \nabla\tilde{\phi}.
\end{equation}
\end{theorem}

\begin{proof} Let $\psi\in C^\infty(M)$ and $(U_t)_{t\in\R}$ be the associated flow of diffeomorphisms of $M$: 
 \begin{equation}\label{eq4.3}
\frac{dU_t(x)}{dt}= \nabla\psi(U_t(x)), \quad x\in M.
\end{equation}
The inverse map $U_t^{-1}$ of $U_t$ satisfies the ODE
\begin{equation}\label{eq4.4}
\frac{dU_t^{-1}(x)}{dt}= -\nabla\psi(U_t^{-1}(x)), \quad x\in M.
\end{equation}
Set $\mu_t=(U_t)_\#\mu$, then $\mu=(U_t^{-1})_\#\mu_t$. Let $\gamma\in\C_o(\sigma,\mu)$ be the optimal 
coupling plan such that 
\begin{equation*}
W_2^2(\sigma,\mu)=\int_{M\times M} d_M^2(x,y)\, d\gamma(x,y).
\end{equation*}
The map $(x,y)\ra (x, U_t(y))$ pushes $\gamma$ forword to a coupling plan $\gamma_t\in \C(\sigma, \mu_t)$. 
Then for $t>0$,
\begin{equation*}
\begin{split}
&\frac{1}{t}\Bigl[ W_2^2(\sigma, \mu_t)-W_2^2(\sigma,\mu)\Bigr]
 \leq \frac{1}{t}\int_{M\times M} \Bigl(d_M^2(x, U_t(y))-d_M^2(x,y)\Bigr)\, d\gamma(x,y)\\
 &=\frac{1}{t}\int_{M\times M} \Bigl(d_M^2(x, U_t(y))- d_M^2(x, \exp_y(t\nabla\psi(y))\Bigr)\,d\gamma(x,y)\\
 &+ \frac{1}{t}\int_{M\times M} \Bigl(d_M^2(x, \exp_y(t\nabla\psi(y)) -d_M^2(x,y)\Bigr)\,d\gamma(x,y)
 =I_1(t)+I_2(t)
 \end{split}
\end{equation*}
respectively. Let $\xi(t)=\exp_x(t\nabla\phi(x))$, by \cite{McCann}, $\xi$ is a minimizing geodesic connecting $x$ and $y=T(x)$.
 By Lemma \ref{lemma4.1}, we have
\begin{equation*}
d_M^2\bigl(x,\exp_y(t\nabla\psi(y)\bigr)-d_M^2(y,x)\leq 2 t\langle \nabla\psi(y), \xi'(1)\rangle_{T_yM}+ o(|t|)\quad\hbox{\rm as}\ t \ra 0.
\end{equation*}
On other hand,
\begin{equation*}
\xi'(1)=d\exp_x(\nabla\phi(x))\cdot\nabla\phi(x)=//_1^{\xi}\nabla\phi(x),
\end{equation*}
where $//_t^\xi$ denotes the parallel translation along the geodesic $\xi$. Hence $|\xi'(1)|=|\nabla\phi(x)|$.
Therefore
\begin{equation*}
I_2(t)\leq 2\int_M \langle \nabla\psi(T(x)), d\exp_x(\nabla\phi(x))\cdot\nabla\phi(x)\rangle\, d\sigma(x)+ o(1)
\end{equation*}

To justifier the passage of limit throught the integral, we note that for $t>0$,
\begin{equation*}
\begin{split}
&\frac{1}{t} \Bigl| d_M^2\bigl(x, \exp_y(t\nabla\psi(y))\bigr) - d_M^2(x,y)\Bigr|\\
&\hskip -6mm \leq \frac{2}{t}\hbox{\rm diam}(M)\,  d_M\bigl(y, \exp_y(t\nabla\psi(y))\bigr)\leq 2\, \hbox{\rm diam}(M)\, |\nabla\psi(y)|.
\end{split}
\end{equation*}
Then
\begin{equation*}
\overline{\lim_{t\downarrow 0}}I_2(t)\leq 2\int_M \langle \nabla\psi(T(x)), d\exp_x(\nabla\phi(x))\cdot\nabla\phi(x)\rangle\, d\sigma(x).
\end{equation*}

For estimating $I_1(t)$, it is obvious that 
\begin{equation}\label{eq4.5}
\lim_{t\downarrow 0} \frac{1}{t}\sup_{y\in M}d_M\bigl(U_t(y), \exp_y(t\nabla\psi(y))\bigr)=0.
\end{equation}
Then $\dis \lim_{t\downarrow 0} I_1(t)=0$. In conclusion:
\begin{equation}\label{eq4.6} 
\overline{\lim_{t\downarrow 0}}\frac{1}{t}\Bigl[ W_2^2(\sigma, \mu_t)-W_2^2(\sigma,\mu)\Bigr]
\leq 2\int_M \langle \nabla\psi(T(x)), d\exp_x(\nabla\phi(x))\cdot\nabla\phi(x)\rangle\, d\sigma(x).
\end{equation}

For obtaining the minoration, we use as in \cite{Villani2} the fact that
$\dis\overline{\lim_{t\downarrow 0}}(-a_t)=-\underline{\lim}_{t\downarrow 0}a_t$.

\vskip 2mm
Let $\tilde\gamma_t\in\C_o(\sigma, \mu_t)$ be the optimal coupling plan: 
\begin{equation*}
W_2^2(\sigma,\mu_t)=\int_{M\times M} d_M^2(x,y)\, d\tilde\gamma_t(x,y).
\end{equation*}
Let $\eta_t\in \C(\sigma, \mu_t)$ be defined by 
\begin{equation*}
\int_{M\times M} f(x,y)d\eta_t(x,y)=\int_{M\times M}f\bigl( x, U_t^{-1}(y)\bigr)\, d\tilde\gamma_t(x,y).
\end{equation*}
Then for $t>0$, 
\begin{equation*}
\frac{1}{t}\Bigl[ W_2^2(\sigma, \mu)-W_2^2(\sigma,\mu_t)\Bigr]
 \leq \frac{1}{t}\int_{M\times M} \Bigl(d_M^2(x, U_t^{-1}(y))-d_M^2(x,y)\Bigr)\, d\tilde\gamma_t(x,y).
\end{equation*}
Let $T_t : M\ra M$ be the optimal transport map which pushes $\sigma$ forword to $\mu_t$, with
$\dis T_t(x)=\exp_x(\nabla\phi_t(x))$. 
As $t\downarrow 0$, the map $T_t$ converges in measure to $T$ (see for example \cite{Villani2}, page 265).
We have
\begin{equation*}
\begin{split}
 & \frac{1}{t}\int_{M\times M} \Bigl(d_M^2(x, U_t^{-1}(y))-d_M^2(x,y)\Bigr)\, d\tilde\gamma_t(x,y)\\
 &=\frac{1}{t}\int_{M} \Bigl(d_M^2(x, U_t^{-1}(T_t(x)))-d_M^2(x,T_t(x))\Bigr)\, d\sigma(x)\\
 &=\frac{1}{t}\int_{M} \Bigl(d_M^2(x, U_t^{-1}(T_t(x)))-d_M^2(x,\exp_{T_t(x)}(-t\nabla\psi(T_t(x)))\Bigr)\, d\sigma(x)\\
 &+\frac{1}{t}\int_{M} \Bigl(d_M^2(x,\exp_{T_t(x)}(-t\nabla\psi(T_t(x)))-d_M^2(x, T_t(x))\Bigr)\, d\sigma(x)
 =J_1(t)+J_2(t)
\end{split}
\end{equation*}
respectively. According to \eqref{eq4.5}, $\lim_{t\downarrow 0} J_1(t)=0$. Concerning $J_2(t)$, we note 
as above,

\begin{equation*}
\begin{split}
&\frac{1}{t}\Bigl|d_M^2\bigl(x,\exp_{T_t(x)}(-t\nabla\psi(T_t(x))\bigr)-d_M^2(x, T_t(x))\Bigr|\\
& \leq \frac{2}{t} \,\hbox{diam}(M) d_M(T_t(x), \exp_{T_t(x)}(-t\nabla\psi(T_t(x)))\\
&\leq 2\,\hbox{\rm diam}(M)\, |\nabla\psi(T_t(x)))|\leq 2\hbox{\rm diam}(M)\, ||\nabla\psi||_\infty.
\end{split}
\end{equation*}

Then by Lemma \ref{lemma4.1}, 
\begin{equation*}
J_2(t)\leq -2\int_M \langle \nabla\psi(T_t(x)), d\exp_x(\nabla\phi_t(x))\cdot\nabla\phi_t(x)\rangle\, d\sigma(x)+ o(1)
\end{equation*}
Therefore
\begin{equation}\label{eq4.7} 
\overline{\lim_{t\downarrow 0}}\frac{1}{t}\Bigl[ W_2^2(\sigma, \mu)-W_2^2(\sigma,\mu_t)\Bigr]
\leq -2\int_M \langle \nabla\psi(T(x)), d\exp_x(\nabla\phi(x))\cdot\nabla\phi(x)\rangle\, d\sigma(x).
\end{equation}
Combining \eqref{eq4.6} and \eqref{eq4.7}, we finally get
\begin{equation}\label{eq4.8}
\lim_{t\downarrow 0}\frac{1}{t}\Bigl[ W_2^2(\sigma, \mu_t)-W_2^2(\sigma,\mu)\Bigr]
= 2\int_M \langle \nabla\psi(T(x)), d\exp_x(\nabla\phi(x))\cdot\nabla\phi(x)\rangle\, d\sigma(x).
\end{equation}
\vskip 2mm

Now if $\mu\in\P_{2, ac}(M)$ and the map $y\ra \exp_y(\nabla\tilde{\phi}(y))$ is the optimal transport map
which pushes $\mu$ to $\sigma$. Consider the minimizing geodesic 
\begin{equation*}
\xi(t)=\exp_y((1-t)\nabla\tilde{\phi}(y)), 
\end{equation*}
which connects $x$ and $y$. We have $\xi'(1)=\nabla\tilde{\phi}(y)$. In this case, replacing 
$d\exp_x(\nabla\phi(x))\cdot\nabla\phi(x)$ in \eqref{eq4.8} by $\nabla\tilde{\phi}(y)$,
 we obtain 
\begin{equation*}\label{eq4.9}
\begin{split}
\lim_{t\downarrow 0}\frac{1}{t}\Bigl[ W_2^2(\sigma, \mu_t)-W_2^2(\sigma,\mu)\Bigr]
&= 2\int_M \langle \nabla\psi(T(x)), \nabla\tilde{\phi}(T(x))\rangle\, d\sigma(x)\\
&=2\int_M \langle \nabla\psi(y), \nabla\tilde{\phi}(y)\rangle\, d\mu(y),
\end{split}
\end{equation*}
from which we get \eqref{eq4.2}. The proof is complete. 
\end{proof}

\section{Parallel translations}\label{sect5}

Before introducing parallel translations on the space $\P_{div}(M)$, let's give a brief review on the definition of parallel translations
on the manifold $M$, endowed with an affine connection. Let $\{\gamma(t);\  t\in [0,1]\}$ be a smooth curve on $M$, 
and $\{Y_t;\ t\in [0,1]\}$ a  family vector fields along $\gamma$: $Y_t\in T_{\gamma(t)}M$. Then there exist vector fields $X$ and $Y$
on $M$ such that 
\begin{equation*}
X(\gamma(t))=\dot\gamma(t),\quad Y(\gamma(t))=Y_t.
\end{equation*}

$Y_t$ is said to be parallel along $\{\gamma(t);\  t\in [0,1]\}$  if 
\begin{equation*}
(\nabla_XY)(\gamma(t))=0,\quad t\in [0,1].
\end{equation*}

\vskip 2mm
Now let $\{c_t;\ t\in [0,1]\}$ be an absolutely curve on $\P_{div}(M)$ such that 
\begin{equation}\label{eq5.1}
\frac{{d}^Ic_t}{dt}=V_{\Phi_t},\quad\hbox{\rm with } \Phi_t\in \D_1^2(c_t).
\end{equation}

Let $\{Y_t;\ t\in [0,1]\}$ be a vector field along $\{c_t;\ t\in [0,1]\}$, that is, $Y_t\in \TT_{c_t}$ given by $Y_t=V_{\Psi_t}$
with $\Psi_t\in\D_1^2(c_t)$.

\begin{theorem}\label{th5.1} Assume that $t\ra c_t$ is $C^1$ in the sense that for any $f\in C^1(M)$, $t\ra F_f(c_t)$ is $C^1$ 
and for  $t\in [0,1]$, $x\ra \Phi_t(x)$ is $C^1$.
If for each $t\in [0,1]$, 
\begin{equation}\label{eq5.2}
|V_{\Phi_t}|_{\TT_{c_t}}^2=\int_M |\nabla\Phi_t(x)|^2\ dc_t(x)>0,
\end{equation}
then there are functions $(\mu, x)\ra \tilde\Phi(\mu, x)$ and $(\mu, x)\ra \tilde\Psi(\mu, x)$ on $\P_2(M)\times M$ 
such that 
\begin{equation}\label{eq5.3}
\tilde\Phi(c_t,x)=\Phi_t(x),\quad \tilde\Psi(c_t, x)=\Psi_t(x); 
\end{equation}
moreover for $x\in M$, $\mu\ra \tilde\Phi(\mu, x)$ and $\mu\ra \tilde\Psi(\mu, x)$ are derivable on $\P_2(M)$ along any constant 
vector fields $V_\psi$, their gradients exist on $\P_{2,ac}(M)$. 
\end{theorem}

\begin{proof}
Fix $t_0\in [0,1]$; consider $\dis\alpha(t)=F_{\Phi_{t_0}}(c_t)$. Then 
\begin{equation*}
\alpha'(t)=\frac{d}{dt}F_{\Phi_{t_0}}(c_t)=\int_M \langle \nabla\Phi_{t_0},  \nabla\Phi_t\rangle\ dc_t,
\end{equation*}
which is $>0$ at $t=t_0$. Therefore there is an open interval $I(t_0)$ of $t_0$ 
such that $t\ra \alpha(t)$ is a $C^1$ diffeomorphism 
from $I(t_0)$ onto an interval  $J(t_0)$ containing $\alpha(t_0)$.
Let $\beta: J(t_0)\ra I(t_0)$ be the inverse map of $\alpha$. We 
have 
\begin{equation*}
F_{\Phi_{t_0}}(c_t)\in J(t_0)\quad \hbox{\rm for } t\in I(t_0).
\end{equation*}

Let 
\begin{equation*}
U(t_0)=\bigl\{\mu\in\P_2(M);\ F_{\Phi_{t_0}}(\mu)\in J(t_0)\bigr\},
\end{equation*}
which is an open set in $\P_2(M)$. Let $r>0$ and $\nu\in \P_2(M)$, we denote by $B(\nu, r)$ the open ball in 
$\P_2(M)$ centered at $\nu$ of radius $r$. Take $r_0>0$ small enough such that 
\begin{equation*}
B(c_{t_0}, r_0)\subset U(t_0). 
\end{equation*}
We define, for $\mu\in B(c_{t_0}, r_0)$, 

\begin{equation}\label{eq5.4}
\tilde\Phi_{t_0}(\mu)=\Phi_{\beta(F_{\Phi_{t_0}}(\mu))}, \quad \tilde\Psi_{t_0}(\mu)=\Psi_{\beta(F_{\Phi_{t_0}}(\mu))}.
\end{equation}
We remark that for $t\in [0,1]$ such that $c_t\in U(t_0)$, we have: $\beta(F_{\Phi_{t_0}}(c_t))=t$. Note that 
$\{c_t;\ t\in [0,1]\}$ is a compact set of $\P_2(M)$ and 

\begin{equation*}
\bigl\{c_t;\ t\in [0,1]\bigr\}\subset \cup_{t_0\in [0,1]} B(c_{t_0}, r_0).
\end{equation*}

There exists a finite number of $t_1, \ldots, t_k\in [0,1]$ such that 
\begin{equation*}
\bigl\{c_t;\ t\in [0,1]\bigr\}\subset \cup_{i=1}^kB(c_{t_i}, r_i).
\end{equation*}

Set $\dis U=\cup_{i=1}^kB(c_{t_i}, r_i)$. Let $\mu\in U$, then $\mu\in B(c_{t_i}, r_i)$; 
according to \eqref{eq5.4}, we define, 

\begin{equation*}
\tilde\Phi_{t_i}(\mu)=\Phi_{\beta_i(F_{\Phi_{t_i}}(\mu))}, \quad \tilde\Psi_{t_i}(\mu)=\Psi_{\beta_i(F_{\Phi_{t_i}}(\mu))}.
\end{equation*}

Then for $t\in [0,1]$ such that $c_t\in B(c_{t_i}, r_i)$,  $\tilde\Phi_{t_i}(c_t)=\Phi_t$ and $\tilde\Psi_{t_i}(c_t)=\Psi_t$.
Now for $r>0$ and $\nu\in\P_2(M)$, we define

\begin{equation}\label{eq5.5}
g_{r,\nu}(\mu)=\exp\Bigl( \frac{1}{W_2^2(\nu, \mu)-r^2}\Bigr),\quad\hbox{if }\ W_2(\nu, \mu)<r,
\end{equation}
and $g_{r,\nu}(\mu)=0$ otherwise. Then $g_{r, \nu}(\mu)>0$ if and only if $\mu\in B(\nu, r)$. By Theorem \ref{th4.2}, 
if $\nu\in\P_{\div}$, $\mu\ra g_{r,\nu}(\mu)$ is derivable along any constant vector field $V_\psi$. Remark that 

\begin{equation*}
\sum_{i=1}^k g_{r_i, c_{t_i}}>0\quad\hbox{\rm on}\ U.
\end{equation*}
Let 
\begin{equation*}
\alpha_i=\frac{g_{r_i, c_{t_i}}}{\sum_{i=1}^k g_{r_i, c_{t_i}}}\quad\hbox{\rm for}\  \mu\in U, 
\quad\hbox{\rm and}\ \alpha_i=0\ \hbox{\rm otherwise}.
\end{equation*}
Now define
\begin{equation*}
\Phi(\mu) =\sum_{i=1}^k \alpha_i(\mu) \tilde\Phi_{t_i}(\mu), \quad \Psi(\mu)=\sum_{i=1}^k \alpha_i(\mu) \tilde\Psi_{t_i}(\mu).
\end{equation*}
We have
\begin{equation*}
\Phi(c_t) =\sum_{i=1}^k \alpha_i(c_t) \tilde\Phi_{t_i}(\mu).
\end{equation*}
Note that $\alpha_i(c_t)>0$ if and only if $c_t\in B(c_{t_i}, r_i)$, which implies that $\tilde\Phi_{t_i}(c_t)=\Phi_t$
and $\dis\Phi(c_t)=\sum_{i=1}^k \alpha_i(c_t)\Phi_t=\Phi_t$. It is the same for $\Psi$. 
The proof is completed.
\end{proof}

Notice that for  such a curve $\{c_t;\ t\in [0,1]\}$ given in Theorem \ref{th5.1}, and $\{Y_t; t\in [0,1]\}$ a vector field along $\{c_t;\ t\in [0,1]\}$
given by $\Psi_t$. If furthermore for any $t\in [0,1]$, $\Psi_t\in H^k(M)$ with $\dis k>\frac{m}{2}+2$, then the extension obtained $\tilde\Psi$ obtained in Theorem \ref{th5.1} satisfies conditions in Theorem \ref{th3.6}. 

\begin{definition}\label{def5.1}
We say that $\{Y_t; t\in [0,1]\}$ is parallel along $\{c_t;\ t\in [0,1]\}$ if 
\begin{equation*}
(\bn_{\frac{d^Ic_t}{dt}}V_{\tilde\Psi})(c_t)=0,\quad t\in [0,1].
\end{equation*}

\end{definition}

\begin{theorem}\label{th5.2} Keeping same notations in Theorem \ref{th5.1}, if $\{Y_t;\ t\in [0,1]\}$ is parallel along 
$\{c_t, t\in [0,1]\}$, the following equation holds
\begin{equation}\label{eq5.6}
\int_M\Bigl \langle \nabla\bigl(\frac{d\Psi_t}{dt}\Bigr)+ \nabla_{\nabla\Phi_t}\nabla\Psi_t,\ \nabla\phi\Bigr\rangle\, dc_t=0,
\quad \phi\in C^\infty(M).
\end{equation}
\end{theorem}

\begin{proof} Note that
\begin{equation*}
(\bD_{\frac{d^Ic_t}{dt}}\tilde\Psi)(c_t)=\frac{d}{dt}\tilde\Psi(c_t)=\frac{d\Psi_t}{dt} \ \hbox{\rm and}\ 
\nabla\tilde{\Psi}(c_t,\cdot)=\nabla\Psi_t.
\end{equation*}
Then \eqref{eq5.6} follows from \eqref{eq3.13}. 
\end{proof}

When $\dis\nabla\bigl(\frac{d\Psi_t}{dt}\Bigr) = \frac{d \nabla\Psi_t}{dt}$,
 it is more convenient to put Equation \eqref{eq5.6} in the following form :
\begin{equation}\label{eq5.7}
\Pi_{c_t}\Bigl(\frac{d}{dt} \nabla\Psi_t+ \nabla_{\nabla\Phi_t}\nabla\Psi_t\Bigr)=0,
\end{equation}
or
\begin{equation}\label{eq5.8}
\frac{d}{dt} \nabla\Psi_t+ \Pi_{c_t}\Bigl(\nabla_{\nabla\Phi_t}\nabla\Psi_t\Bigr)=0,
\end{equation}
where  $\dis \Pi_{c_t}$ the orthogonal projection from $L^2(M, TM, c_t)$ onto $\TT_{c_t}$. By arguments in the proof
of Proposition \ref{prop3.2}, when $dc_t=\rho_t\,dx$ with $\rho_t\in C^2(M)$ and $\rho_t>0$, $\dis \Pi_{c_t}$ admits the expression

\begin{equation*}
\Pi_{c_t}u=(\nabla \L_{c_t}^{-1} \div_{c_t})(u),\quad u\in L^2(M, TM, c_t).
\end{equation*}
The price for this pointwise formulation of \eqref{eq5.7} as well as of \eqref{eq5.8} is the involement of second order derivative
of $\Psi$.

\begin{remark} {\rm Let $s\ra \xi(s)$ is a smooth curve on $M$ such that $\xi(0)=x$ and $\xi'(0)=\nabla\Phi_t(x)$, then
\begin{equation}\label{eq5.9}
\frac{d}{dt} \nabla\Psi_t+ \nabla_{\nabla\Phi_t}\nabla\Psi_t
=\lim_{\eps \ra 0}\frac{\tau_\eps^{-1}\nabla\Psi_{t+\eps}(\xi(\eps))-\nabla\Psi_t(x)}{\eps},
\end{equation}
where $\tau_s$ is the parallel translation along $s\ra \xi(s)$. We refind the similar expression of parallel translations given in \cite{AG}.
}
\end{remark}

\begin{proposition}\label{prop5.5} Assume that the curve $\{c_t;\ t\in [0,1]\}$ is induced by a flow of diffeomorphisms $\Phi_t$,
 that is, there is a $C^{1,2}$ function $(t, x)\ra \Phi_t(x)$ such that 

\begin{equation*}
\left\{ \begin{array}{ccc} \frac{dU_{s,t}(x)}{dt}&=&\nabla\Phi_t(U_{s,t}(x)),\quad U_{s,s}(x)=x,\\
c_t&=&(U_{0,t})_\# c_0.
\end{array}\right.
\end{equation*}
Then for any $u_0=\nabla\Psi_0\in \TT_{c_0}$, there is a unique vector field $u_t=\nabla\Psi_t\in \TT_{c_t}$ along 
$\{c_t; t\in [0,1]\}$ such that

\begin{equation}\label{eq5.10}
\Pi_{c_t}\Bigl(\lim_{\eps \ra 0}\frac{\tau_\eps^{-1}\nabla\Psi_{t+\eps}(U_{t, t+\eps}(x))-\nabla\Psi_t(x)}{\eps}\Bigr)=0
\end{equation}
holds in $L^2(c_t)$, where $\tau_\eps$ is the parallel translation along $\{s\ra U_{t,t+s}(x), s\in [0, \eps]\}$. 
\end{proposition}

\vskip 2mm

\begin{proof} Following Section 5 of \cite{AG}, for $s\leq t$, we define 
\begin{equation*}
{\mathcal P}_{t,s}: \TT_{c_s}\ra \TT_{c_t},\quad u_s\ra \Pi_{c_t}\bigl( \tau_{t-s}u_s\circ U_{s,t}^{-1}\bigr).
\end{equation*}
For a subdivision ${\mathcal D}=\{0=t_0<t_1<\ldots<t_n=1\}$ of $[0,1]$, we define 
\begin{equation*}
{\mathcal P}_{\mathcal D}: \TT_{c_0}\ra \TT_{c_1},\quad u_0\ra ({\mathcal P}_{1,t_{n-1}}\circ\cdots\circ {\mathcal P}_{t_1,0})(u_0).
\end{equation*}
Under the assumption of Theorem, we have the uniform bound
\begin{equation*}
\sup_{(t,x)\in [0,1]\times M} ||\nabla^2\Phi_t(x)|| <+\infty,
\end{equation*}
 which allows us to mimic the construction of section 5 in \cite{AG}, 
 so that we get that ${\mathcal P}_{\mathcal D}$ converges as ${\mathcal D}$
 becomes finer and finer, with $|{\mathcal D}|=\max_{i} |t_i-t_{i-1}|\ra 0$.
\end{proof}

\vskip 2mm
As a result of \eqref{eq5.10}, we have as in \cite{AG} the following property: 

\begin{proposition}\label{prop5.6} Let $\{\nabla\Psi_t; t\in [0,1]\}$ be given in Proposition \ref{prop5.5}, then
\begin{equation}\label{eq5.11}
\frac{d}{dt}||\nabla\Psi_t||_{c_t}^2 =0.
\end{equation}
\end{proposition}

\begin{proof} We have $c_{t+\eps}=(U_{t, t+\eps})_\#c_t$, and
\begin{equation*}
\int_M|\nabla\Psi_{t+\eps}(x)|^2\, dc_{t+\eps}(x)=\int_M |\nabla\Psi_{t+\eps}(U_{t,t+\eps}(x))|^2\, dc_t(x).
\end{equation*}
Therefore

\begin{equation*}
\begin{split}
||u_{t+\eps}||_{\TT_{t+\eps}}^2-||u_t||_{\TT_{c_t}}^2
&=\int_M \Bigl[|\tau_\eps^{-1}\nabla\Psi_{t+\eps}(U_{t, t+\eps}(x))|^2- |\nabla\Psi_t(x)|^2\Bigr]\, dc_t(x)\\
&\hskip -6mm =\int_M \Bigl\langle \tau_\eps^{-1}\nabla\Psi_{t+\eps}(U_{t, t+\eps}(x))-\nabla\Psi_t(x),
 \tau_\eps^{-1}\nabla\Psi_{t+\eps}(U_{t, t+\eps}(x))\Bigr\rangle\, dc_t(x)\\
 &\hskip -6mm + \int_M \Bigl\langle \nabla\Psi_t(x),\ \tau_\eps^{-1}\nabla\Psi_{t+\eps}(U_{t, t+\eps}(x))-\nabla\Psi_t(x)
\Bigr\rangle\, dc_t(x).
\end{split}
\end{equation*}
It follows that 
\begin{equation*}
\frac{d}{dt}||\nabla\Phi_t||_{c_t}^2 =2\int_M \Bigl\langle \lim_{\eps \ra 0}\frac{\tau_\eps^{-1}\nabla\Psi_{t+\eps}(U_{t, t+\eps}(x))-\nabla\Psi_t(x)}{\eps}, \nabla\Psi_t(x) \Bigr\rangle\, dc_t(x)=0.
\end{equation*}

\end{proof}

\vskip 2mm
In what follows, we will relaxe a bit conditions in Proposition \ref{prop5.5}. We return to the situation in Theorem \ref{th5.1}. 
Let $\{c_t; t\in [0,1]\}$ be an absolutely curve in $\P_{\div}(M)$ satisfying conditions in Theorem \ref{th5.1}, set 
\begin{equation*}
\frac{d^Ic_t}{dt}=V_{\Phi_t}.
\end{equation*}
If furthermore $(t,x)\ra \nabla^2\Phi_t(x)$ is continuous, according to the the construction, the extension $\tilde\Phi(\mu, x)$ 
of  $(t,x)\ra \nabla^2\Phi_t(x)$ obtained in \eqref{eq5.3} satisfies $(\mu, x)\ra \nabla^2\tilde\Phi(\mu, x)$ is continuous. In particular,
the condition \eqref{eq2.4}

\begin{equation*}
\sup_{(\mu, x)\in \P_2(M)\times M}||\nabla^2\tilde\Phi(\mu, x)||<+\infty,
\end{equation*}
holds. By theorem \ref{th2.2}, there exists a solution $(\mu_t, U_t)$ to the following Mckean-Vlasov equation

\begin{equation*}
 \frac{dU_{t}(x)}{dt}=\nabla\tilde\Phi(\mu_t, U_{t}(x)),\quad U_{0}(x)=x,
\end{equation*}
with $\dis \mu_t=(U_t)_\# c_0$ which solves the ODE on $\P_2(M)$:
\begin{equation}\label{eq5.12}
\frac{d^I\mu_t}{dt}=V_{{\tilde\Phi}(\mu_t,\cdot)}.
\end{equation}

\begin{theorem} If the ODE \eqref{eq5.12} has the unique solution, then for each $V_{\Psi_0}\in \TT_{c_0}$, 
there is a vector field $\{V_{\Psi_t}\in \TT_{c_t};\ t\in [0,1]\}$ along $\{c_t;\ t\in [0,1]\}$ such that 
\begin{equation*}
\frac{d}{dt}||\nabla\Psi_t||_{c_t}^2 =0.
\end{equation*}
holds in $L^2(c_t)$.

\end{theorem}

\begin{proof}  Note that $\dis \nabla\tilde\Phi(c_t, x)=\nabla\Phi_t(x)$, then $\dis V_{\tilde\Phi(c_t,\cdot)}=V_{\Phi_t}$.
The curve $\{c_t; \ t\in [0,1]\}$ is therefore a solution to
\begin{equation*}
\frac{d^Ic_t}{dt}=V_{\Phi_t}=V_{\tilde\Phi(c_t,\cdot)}.
\end{equation*}
Under the assumption of uniqueness of solution to \eqref{eq5.12}, we get that $c_t=\mu_t$ for $t\in [0, 1]$. Now 
by arguments in the proof of Propositions \ref{prop5.5} and \ref{prop5.6}, we obtain the result.

\end{proof}

\begin{remark}
The parallel translations along diffusion paths on the Wasserstein space are discussed in a forthcoming paper \cite{DingFang}.
\end{remark}

\vskip 4mm
 {\bf Acknowledgement:} This work has been prepared in a joint PhD program between the
  Institute of Applied Mathematics, Academy of Mathematics and Systems Science (Beijing, China) and 
  the Institute of Mathematics of Burgundy, University of Burgundy (Dijon, France), 
  the first named author is grateful to the hospitality of these two institutions, the financial support of 
  China Scholarship Council is particularly acknowledged.


\begin{thebibliography}{99}

\bibitem{AG} L. Ambrosio and N. Gigli, Construction of the parallel transport in the Wasserstein space. 
{\it Methods Appl. Anal. }15 (2008), no. 1, 1?29.

\bibitem{AGS} L. Ambrosio, N. Gigli and G. Savar\'e, {\it Gradient flows in metric spaces and in the space of probability measures}, Lect. in Math.,
ETH Z\"urich, Birkh\"auser Verlag, Basel, 2005.

\bibitem{BakryEmery} D. Bakry and M. Emery, Diffusion hypercontractivities, {\it S\'em. de Probab.}, XIX, Lect. Notes in Math., 1123 (1985), 177-206, Springer.

\bibitem{BB} J.D. Benamou and Y. Brenier: A computational fluid mechanics
solution to the Monge-Kantorovich mass transfer problem, {\it Numer.
Math.}, {\bf 84} (2000), 375-393.

\bibitem{BLPR} R. Buckdahn, J. Li, S. Peng, C. Rainer, Mean-field stochastic differential equations and associated PDEs,
 {\it Ann. Prob.} , 45 (2017), 824-878.

\bibitem{Brenier} Y. Brenier, Polar factorization and monotone rearrangement of vector valued functions, {\it Comm. Pure Appl. Math. }, 44 (1991), 375-417.




\bibitem{ABC} A.B. Cruzeiro: Equations diff\'erentielles sur l'espace de
Wiener et formules de Cameron-Martin non lin\'eaires,  {\it J. Funct.
Analysis}, {54} (1983), 206-227.

\bibitem{DingFang} Hao Ding and S. Fang, Stochastic parallel translations on the Wasserstein space, {\it in preparation}.



\bibitem{FangShao} S. Fang and J. Shao, Fokker-Planck equation with respect to heat measures on loop groups {\it Bull. Sci. Math. }, 135 (2011), 775-794.

\bibitem{Gigli} N. Gigli,  On the inverse implication of Brenier-McCann theorems and the structure of $(\P_2(M), W_2)$.
{\it  Methods Appl. Anal.} 18 (2011), no. 2, 127?158. 



\bibitem{Kunita} H. Kunita, {\it Stochastic Flows and Stochastic Differentail Equations}. Cambridge University Press, 1990.

\bibitem{LiLi} Songzi Li and Xiangdong Li, W -entropy formulas  and Langevin deformation of flows
on the Wasserstein space over Riemannian manifolds, arXiv:1604.02596v1 (58 pages).

\bibitem{LiLi2} Songzi Li and Xiangdong Li, W -entropy formulas on super Ricci flows and Langevin deformation on Wasserstein space over Riemannian manifolds {\it Sci. China Math.}, 61 (2018), 1385-1406.

\bibitem{Lott} J. Lott, Some geometric calculation on Wasserstein space, {\it Commun. Math. Phys.}, 277 (2008), 423-437.

\bibitem{LottVillani}  J. Lott and C. Villani: Ricci curvature for metric-measure spaces via optimal transport, {\it Ann of Math.}, 169 (2009), 903-991.

\bibitem{McCann} R. McCann, Polar factorization of maps on Riemannian manifolds, {\it Geo. Funct. Anal.}, 11 (2001), 589-608.

\bibitem{Malliavin} P. Malliavin, Stochastic analysis, Grund. Math. Wissen., vol. 313, Springer, 1997.

\bibitem{Otto} F. Otto: The geometry of dissipative evolution equations: The
porous medium equation, Comm. partial Diff. equations, {\bf 26}
(2001), 101-174.

\bibitem{OV} F. Otto and Villani, \emph{Generalization of an inequality by
Talagrand and links with the logarithmic Sobolev inequality,}
J. Funct. Anal. 173(2000), 361-400.



\bibitem{Sturm} K. T. Sturm, On the geometry of metric measure spaces, {\it Acta Math.}, 196 (2006), 65-131.

\bibitem{SV} K.T. Sturm, M.K. Von Renesse, Transport inequalities, gradient estimates, entropy and Ricci curvature, {\it Comm. Pures Appl. Math.}, 58 (2005), 923-940.

\bibitem{Villani1} C. Villani, {\it Optimal transport, Old and New}, vol. 338, Grund. Math. Wiss., Springer-Verlag, Berlin, 2009.

\bibitem{Villani2} C. Villani, {\it Topics in optimal transportation}, Graduate Studies in Mathematics, 58 (2003), AMS,
 Providence, Ehode Island.
 
 \bibitem{Wang} Feng-Yu Wang, Diffusions and PDEs on Wasserstein Space, {\it arXive: 1903.02148v2}, 2019.

\end{thebibliography}
\end{document}